%% file: main-arxiv.tex
\documentclass[11pt]{article}

\usepackage[letterpaper,margin=0.99in]{geometry}

\newcommand\blfootnotea[1]{%
  \begingroup
  \renewcommand\thefootnote{}\footnote{#1}%
  \endgroup
}

\usepackage[dvipsnames,table,xcdraw]{xcolor}
\usepackage[utf8]{inputenc}

\usepackage[T1]{fontenc}    %
\usepackage{url}      %
\usepackage{booktabs}     %
\usepackage{nicefrac}     %

\usepackage{amsthm,amsfonts,amsmath,amssymb,epsfig,color,float,graphicx,verbatim,
enumitem}

\usepackage{algpseudocode,algorithm,algorithmicx}  

\usepackage{bbm}
\usepackage{caption}
\usepackage{turnstile}

\usepackage[
backend=biber,
style=alphabetic,
maxbibnames=15,
maxalphanames=10,
minalphanames=6,
doi=false,
isbn=false,
url=true,
eprint=false,
backref=true,
]{biblatex}
\DefineBibliographyStrings{english}{
  backrefpage  = {},
  backrefpages = {},
}

\DeclareFieldFormat{bracketswithperiod}{\mkbibbrackets{#1}}

\renewbibmacro*{pageref}{%
  \iflistundef{pageref}
    {}
    {\printtext[bracketswithperiod]{%
       \ifnumgreater{\value{pageref}}{1}
         {\bibstring{backrefpages}}
         {\bibstring{backrefpage}}%
       \printlist[pageref][-\value{listtotal}]{pageref}}%
     }
 }
     
\renewbibmacro{in:}{}
\AtEveryBibitem{\clearfield{pages}}

\definecolor{lb}{RGB}{0, 100, 200}
\definecolor{green2}{RGB}{60, 120, 0}

\usepackage[colorlinks,citecolor=green2,linkcolor=lb,bookmarks=true]{hyperref}

\usepackage[nameinlink,capitalise]{cleveref}

\usepackage{thm-restate}
\usepackage{fontawesome}

\newlist{itemizec}{itemize}{2}
\setlist[itemizec,1]{label=\faCaretRight ,wide, parsep= 0.05pt, left = 11pt}

\def\E{\mathbb E}
\def\pE{\widetilde{\mathbb E}}

\def\R{\mathbb R}

\def\N{\mathbb N}

\def\eps{\epsilon}

\let\vec\mathbf

\newcommand{\bA}{\vec{A}}

\newcommand{\bI}{\vec{I}}

\newcommand{\poly}{\mathrm{poly}}
\newcommand{\polylog}{\mathrm{polylog}}
\newcommand{\cN}{\mathcal{N}}
\newcommand{\cL}{\mathcal{L}}
\newcommand{\cT}{\mathcal{T}}

\newcommand{\cF}{\mathcal{F}}

\newcommand{\cY}{\mathcal{Y}}

\newcommand{\cA}{\mathcal{A}}
\newcommand{\cC}{\mathcal{C}}

\newcommand{\cD}{\mathcal{D}}

\newcommand{\op}{\mathrm{op}}

\crefformat{equation}{(#2#1#3)}
\crefname{equation}{Equation}{Equations}
\crefname{lemma}{Lemma}{Lemmata}
\crefname{claim}{Claim}{Claims}
\crefname{fact}{Fact}{Facts}
\crefname{theorem}{Theorem}{Theorems}
\crefname{proposition}{Proposition}{Propositions}
\crefname{corollary}{Corollary}{Corollaries}
\crefname{remark}{Remark}{Remarks}
\crefname{definition}{Definition}{Definitions}
\crefname{question}{Question}{Questions}
\crefname{condition}{Condition}{Conditions}
\crefname{figure}{Figure}{Figures}

\newtheorem{theorem}{Theorem}[section]

\newtheorem{claim}[theorem]{Claim}

\theoremstyle{definition}
\newtheorem{fact}[theorem]{Fact}
\newtheorem{definition}[theorem]{Definition}

\newtheorem{question}[theorem]{Question}

\newtheorem{remark}[theorem]{Remark}

\allowdisplaybreaks

\theoremstyle{definition}

\usepackage{color}
\definecolor{Red}{rgb}{1,0,0}
\definecolor{Blue}{rgb}{0,0,1}
\definecolor{DGreen}{rgb}{0,0.55,0}
\definecolor{Purple}{rgb}{.75,0,.25}
\definecolor{Grey}{rgb}{.5,.5,.5}

\usepackage{turnstile}

\newcommand{\iprod}[1]{\langle#1\rangle}

\newcommand{\Paren}[1]{\left(#1\right)}

\newcommand{\Brac}[1]{\left[#1\right]}

\date{}

\title{SoS Certifiability of Subgaussian Distributions\\  and its Algorithmic Applications
\blfootnotea{Authors are listed in alphabetical order.}
}

\begin{document}
\maketitle
\vspace{-0.4in}
\begin{center}
\renewcommand*{\thefootnote}{\fnsymbol{footnote}}
\begin{tabular}{ccc}
		{\large{Ilias Diakonikolas}}\footnotemark[2] & \hspace*{.1in} & {\large{Samuel B. Hopkins}}\footnotemark[3] \\
		{\large{UW Madison}} & \hspace*{.1in} & {\large{MIT}} \\
		{\large{\texttt{ilias@cs.wisc.edu}}} &  & {\large{\texttt{samhop@mit.edu}}} \vspace{.25in}\\
		{\large{Ankit Pensia}}\footnotemark[4]  & \hspace*{.1in}& {\large{Stefan Tiegel}}\footnotemark[5] \\
		{\large{Simons Institute, UC Berkeley}} & \hspace*{.1in}  & {\large{ETH Z\"urich}} \\
		{\large{\texttt{ankitp@berkeley.edu}}}  & \hspace*{.1in}& {\large{\texttt{stefan.tiegel@inf.ethz.ch}}} \\
			\end{tabular}
\footnotetext[2]{Supported by NSF Medium Award CCF-2107079 and an H.I. Romnes Faculty Fellowship.}
\footnotetext[3]{Supported by NSF CAREER award no. 2238080 and MLA@CSAIL.}
\footnotetext[4]{Most of this work was done while the author was supported by IBM Herman Goldstine Fellowship. }
\footnotetext[5]{Supported by the European Union’s Horizon research and innovation programme (grant agreement no. 815464). Most of this work was done while the author was visiting MIT.}

\vspace{0.4in}
\today
\end{center}

\renewcommand*{\thefootnote}{\arabic{footnote}}
\setcounter{footnote}{0}

\thispagestyle{empty}
\vspace{0.4in}

\begin{abstract}
We prove that there is a universal constant $C>0$ so that for every $d \in \N$, 
every centered subgaussian distribution $\cD$ on $\R^d$, and every even $p \in \N$, 
the $d$-variate polynomial 
$(Cp)^{p/2} \cdot \|v\|_{2}^p - \E_{X \sim \cD} \iprod{v,X}^p$
is a sum of square polynomials.
This establishes that every subgaussian 
distribution is \emph{SoS-certifiably subgaussian}---a condition 
that yields efficient learning algorithms for a wide variety of 
high-dimensional statistical tasks. %
As a direct corollary, we obtain computationally efficient algorithms with near-optimal guarantees
for the following tasks, when given samples from an arbitrary subgaussian distribution: 
robust mean estimation, list-decodable mean estimation, 
clustering mean-separated mixture models, robust covariance-aware mean estimation, robust covariance estimation, 
and robust linear regression.
Our proof makes essential use of Talagrand's generic chaining/majorizing measures theorem.
\end{abstract}

\newpage
\setcounter{page}{1}

\section{Introduction} %
\label{sec:introduction}

\paragraph{Motivation: Distributional Assumptions {\em vs} Accuracy in Robust Statistics} 
Algorithmic robust statistics aims to design computationally efficient 
estimators that achieve (near-)optimal accuracy in the presence of a 
significant fraction of corrupted data. The prototypical problem 
in this field is {\em robust mean estimation}: 
the input is a multiset $S$ consisting of $n$ points in $\R^d$ 
and a contamination parameter $\eps \in (0, 1/2)$ such that, roughly speaking, an 
unknown $(1-\eps)$-fraction of the points in $S$ are drawn from 
an unknown distribution $P$ in a known family $\cal{D}$.
Importantly, no 
assumptions are made on the remaining $\eps$-fraction of the points; they 
could be arbitrary or even adversarially selected.
The goal of the 
learner is to estimate the unknown mean $\mu$ of $P$. 

While the statistical limits of this basic task were 
well-understood early on for a range of distribution families 
$\cal{D}$~\cite{DonLiu88a,HubRon09}, our understanding of its 
{\em computational complexity} has 
developed only fairly recently~\cite{DiaKKLMS16-focs, LaiRV16,DiaKan22-book}. 
Specifically, for the textbook setting that $\cal D$ is the class of 
Gaussian distributions,~\cite{DiaKKLMS16-focs} gave a polynomial sample 
and time algorithm that approximates the target mean within $\ell_2$-error 
$\tilde{O}(\eps)$. This upper bound nearly matches the information-theoretic 
limit of $\Theta(\eps)$ on the best possible error, and is 
constant factor optimal for efficient algorithms in broad (yet 
restricted) models of computation~\cite{DiaKS17,BreBHLS21,DiaKPP24-sos}.

\looseness=-1The Gaussian assumption on the clean data is typically insufficient to 
accurately model a range of realistic applications. Motivated by 
this limitation, a flurry of research in algorithmic 
robust statistics has pursued efficient learning algorithms for 
more general distribution families. One of the most general 
assumptions in the literature posits that $\cal D$ is 
the class of bounded covariance distributions 
(specifically, that for each $P \in \cal{D}$ its 
covariance $\vec\Sigma_P$ is spectrally bounded by the identity). For this class of 
distributions, one can efficiently and robustly approximate the mean within 
$\ell_2$-error of $O(\eps^{1/2})$, which is information-theoretically optimal up 
to a constant factor~\cite{DiaKKLMS17, SteCV18}. 

While the latter error guarantee is best possible for worst-case bounded 
covariance distributions, 
it would be preferable to have an error dependence closer to 
$\tilde{O}(\eps)$---the bound achievable in the Gaussian case. This leads to the following natural question:
\begin{center}
{\em Can we obtain computationally efficient robust mean estimators\\ 
for broad classes of distributions with error $O(\eps^{1/2+c})$, 
for some constant $c>0$?}
\end{center}
\looseness=-1Analogous algorithmic questions can be posed for a number of high-dimensional estimation tasks, including robust linear regression and covariance estimation, 
clustering of mixture models, and list-decodable learning. 
See~\Cref{table:summary,sec:apps-detailed} for concrete algorithmic applications 
of the main result obtained in this work to a range of high-dimensional estimation tasks. 

\looseness=-1
\paragraph{Robust Estimation \& {\em Subgaussian} Distributions}
A natural nonparametric family broadly generalizing 
the Gaussian distribution---for which a qualitatively similar error guarantee 
is {\em information-theoretically} possible\footnote{Specifically, for the class of subgaussian distributions, the mean can be robustly estimated within error $O(\eps\sqrt{\log(1/\eps)})$, which is information-theoretically optimal within constant factors (see, e.g.,~Chapter~1 of \cite{DiaKan22-book}). On the other hand, known computational lower bounds~\cite{DiaKKPP22-colt} suggest that obtaining error $O(\eps^{1-c})$ requires complexity $d^{\Omega(1/c)}$. See \Cref{rem:tradeoffs-robust-mean} for a more detailed explanation.}---is the class of {\em subgaussian} distributions.  
A distribution is subgaussian if all its linear projections 
have tail probabilities decaying at least as fast as Gaussian tails. 
For our purposes, an equivalent definition in terms of moments rather than tail bounds is appropriate:
\begin{definition}[Subgaussian distributions; see, e.g.,~{\cite{Vershynin18}}] \label{def:sg}
    For $s>0$, a distribution $P$ over $\R^d$ with mean $\mu$ is 
    $s$-subgaussian if for all $v \in \R^d$,
    \begin{align*}
       \forall m\geq 1:  \left( \E_{X \sim P}[|\langle v, X - \mu \rangle|^m] \right)^{1/m}
        \leq C s \sqrt{m}\|v\|_2\,,   
    \end{align*}    
    where the universal constant $C>0$ is chosen so that $\cN(0,\bI_d)$ is $1$-subgaussian.
For convenience, we will henceforth call $O(1)$-subgaussian distributions simply ``subgaussian distributions''.
\end{definition}

Subgaussian distributions~\cite{Kahane1960} are widely studied in  
computer science and statistics; see, 
e.g.,~\cite{LedTal91,BouLM13,Vershynin18,Wainwright19, Talagrand21}.  
Subgaussianity captures a core structural property 
of ``reasonable'' data---decaying tails of linear projections---that 
is particularly useful in the context of statistical estimation, 
while amounting to a significantly weaker assumption than Gaussianity.
As alluded to above, for any subgaussian distribution,
it is information-theoretically possible to robustly 
estimate the mean within error $\tilde{O}(\eps)$. 
More generally, for a range of robust estimation tasks, 
the information-theoretically optimal error 
depends mainly on the probabilities of tail events of linear projections.
For such tasks, optimal estimation rates for Gaussians and subgaussian distributions 
are usually very close to each other (see \Cref{sec:apps-detailed} for concrete examples).

\looseness=-1The main challenge is determining 
whether there exist \emph{computationally efficient} %
algorithms that can achieve the improved error rates 
which are information-theoretically possible under the subgaussian assumption.
Concretely, prior to our work, the best known error guarantee 
for robust mean estimation of general subgaussian distributions was $O(\eps^{1/2})$, 
i.e., the same as for the much broader class of 
bounded covariance distributions---even though 
error $\tilde{O}(\eps)$ is possible in exponential time. 
As stated above, for the very special case of the Gaussian distribution, 
a polynomial time algorithm with $\tilde{O}(\eps)$ error 
was previously known~\cite{DiaKKLMS16-focs}.

\paragraph{\emph{Certifiable} Subgaussianity and Sum-of-Squares Proofs} 
A natural way to move forward algorithmically is to consider 
a ``computationally friendly'' notion of a subgaussian distribution.
By \Cref{def:sg}, the moments of every linear projection of a subgaussian distribution 
are bounded. We say that a distribution is {\em certifiably subgaussian} 
if this moment-boundedness in all directions has a short certificate 
in the form of a \emph{sum of squares proof}---in particular, if
the multivariate polynomials $p_m(v) = (C \sqrt{m}\|v\|_2)^m - \E[\iprod{X - \mu,v}^m]$ 
are sums of squares for all even $m$.

Certifiably subgaussian distributions form a subclass %
of subgaussian distributions.
Importantly, the improved error rates attainable 
for subgaussians can typically be obtained by 
{\em efficient algorithms} under the stronger assumption 
of certifiable subgaussianity, via the sum-of-squares method.

To set up our main result, we need a formal definition of certifiable subgaussianity.
To start, we say that a multivariate distribution $P$ over $\R^{d}$ with mean $\mu$ has 
$(B_m,m)$-bounded moments, or is $(B_m,m)$-bounded, 
if for all vectors $v$ it satisfies 
$\E_{X \sim P}\left[\left\langle X - \mu, v\right\rangle^m\right] \leq B_m^m\|v\|_2^m$.
Our first definition requires that this inequality has an SoS proof.
\begin{definition}[Certifiably Bounded Distributions~\cite{KotSte17, HopLi18}]
\label{def:cert-bdd}
Let $m \in \N$ be even and $B_m > 0$. 
We say that a distribution $P$ over $\R^{d}$ with mean $\mu$ is 
$(B_m,m)$-certifiably bounded if the polynomial 
$p(v) := B_m^m\|v\|_2^m - \E_{X \sim P}\left[\left\langle X - \mu, v\right\rangle^m\right]$ is a sum of squares polynomial in $v$.%
\footnote{We remark that the definitions in \cite{KotSte17,HopLi18} also allowed 
unit norm constraints on the variable $v$. That is, a distribution $P$ was termed  
certifiably bounded if the polynomial 
$p(v) := B_m^m\|v\|_2^m - \E_{X \sim P}\left[\left\langle X - \mu, v\right\rangle^m\right]$ 
satisfies the identity $p(v) = h(v) (\|v\|_2^2-1) +\sum_i q_i^2(v)$, for 
some polynomials $h(\cdot)$, $q_1(\cdot),\dots,q_m(\cdot)$ that 
are each of degree $O(m)$.
Equivalently, in the notation defined later, 
$\left\{\sum_{i=1}^d v_i^2 = 1\right\} \sststile{O(m)}{v}  p(v) \geq 0$. 
In contrast, \Cref{def:cert-bdd} is stricter and requires that 
$ \sststile{}{v}  p(v) \geq 0$, which is equivalent to $\sststile{m}{v}  p(v) \geq 0$.

See \Cref{rem:unit-norm} for further details. }
\end{definition}

{A distribution is called {\em certifiably subgaussian} if it 
satisfies \Cref{def:cert-bdd} for 
the appropriate value of the parameter $B_m$ 
corresponding to the moment bounds of subgaussian distributions.}

\begin{definition}[Certifiably Subgaussian Distributions~\cite{KotSte17, HopLi18}] 
\label{def:cert-sg}
For $s>0$, we say that a distribution $P$ 
is $s$-certifiably subgaussian if it is 
$(Cs\sqrt{m}, 2m)$-certifiably bounded for all even $m \in \N$, 
where $C$ is the universal constant from \Cref{def:sg}.
If $s= O(1)$, we usually omit the parameter $s$ and 
call such distributions ``certifiably subgaussian''.
\end{definition}

The notion of certifiability naturally arises while designing 
efficient algorithms that access higher-order moments of the 
samples. 
For example, consider the problem of computing the maximum of
$\|\bA v\|_{\ell_m}^m := \sum_{i=1}^n (a_i^\top v)^m$ over unit vectors $v\in\R^d$, 
known as the ``$2$-to-$m$-norm,'' of $\bA$, or the tensor injective norm 
of the $m$-tensor $\sum_{i=1}^n a_i^{\otimes m}$. 
It turns out that this problem arises frequently 
when constructing robust estimation algorithms.

While computing the injective tensor norm is believed to be computationally hard 
(even to approximate) in the worst-case (over $a_i$'s) 
for $m \geq 4$~\cite{BarBHKSZ12},
it is polynomial-time approximable 
in some average-case settings:
specifically, if each row $a_i$ is sampled i.i.d.\ from an isotropic certifiably bounded distribution, 
then the sum-of-squares method succeeds with high probability.

A recent line of work, starting with~\cite{KotSte17,HopLi18}, showed how to 
leverage the Sums-of-Squares ``proofs to algorithms'' approach in order to design
significantly more accurate polynomial-time robust mean estimators 
for distributions that satisfy~\Cref{def:cert-bdd}. Specifically, 
these works gave SDP-based algorithms to achieve error of $O{_k}(\eps^{1-1/k})$, 
for any $k \geq 2$, with sample and computational complexities $(d/\eps)^{\poly(k)}$. 
In particular, this implies that one can achieve error $O(\eps^{1/2+c})$, 
where $c \in (0, 1/2)$ is a universal constant, in $\poly_{{1/2 - c}}(d/\eps)$ time. 
Moreover, as explained in \Cref{rem:tradeoffs-robust-mean} below, 
there is evidence that this tradeoff between computational complexity and error guarantee 
is essentially best possible, even for simple explicit classes of subgaussian distributions. 

\begin{remark}[Information-Computation Tradeoffs for Robust Mean Estimation] \label{rem:tradeoffs-robust-mean}
While the information-theoretically optimal error rate for robustly estimating 
the mean of a subgaussian distribution is $\tilde{\Theta}(\eps)$ 
and can be achieved with ${O}(d/\eps)$ many samples, 
\cite{DiaKKPP22-colt} gave SQ lower bounds 
providing evidence that any computationally efficient algorithm for this task 
with error $O(\eps^{1-1/t})$, for $t \geq 2$, requires sample complexity $d^{\Omega(t)}$. 
This lower bound also applies to low-degree polynomial tests (via~\cite{BreBHLS21}) 
and to SoS algorithms~\cite{DiaKPP24-sos}.
\end{remark}

Beyond robust mean estimation, the approach 
of~\cite{KotSte17,HopLi18} has led to qualitatively improved 
efficient algorithms for a range of estimation tasks involving certifiably 
subgaussian distributions, 
including robust linear regression~\cite{KliKM18,BakPra21}, 
robust covariance estimation~\cite{KotSte17, KotSS18}, 
robust sparse estimation~\cite{DiaKKPP22-colt}, 
clustering mixture models (even without outliers) under near-optimal 
pairwise separation~\cite{HopLi18, KotSteinhardt17}, 
list-decodable mean estimation~\cite{KotSteinhardt17, DiaKKPP22-neurips}, 
learning under privacy constraints~\cite{KotMV21}, and 
testable learning~\cite{GolKSV23}. Establishing the certifiability 
of (a larger subclass of) subgaussian distributions directly expands 
the reach of known algorithms for these estimation tasks.

\paragraph{This Work: Characterization of Certifiable Subgaussianity?}

The aforementioned progress notwithstanding, \Cref{def:cert-sg} is somewhat unsatisfying for the 
following reason.
On the one hand, certifiably subgaussian distributions are by definition
a subset of subgaussian distributions.
On the other hand, it is a priori unclear which distributions satisfy \Cref{def:cert-sg}. 
Prior work observed that 
product and rotationally invariant subgaussian distributions (and linear transformations thereof) are certifiably subgaussian.
Moreover, 
\cite{KotSteinhardt17} showed that Poincare distributions are certifiably bounded.
However, their proof approach inherently fails for the class of subgaussians.\footnote{In particular, 
their proof relies on the so-called $(m,\sigma)$-moment property from \cite[Question 1.7]{KotSteinhardt17}, 
which is satisfied by Poincare distributions, with the resulting certifiable moment $B_m$ 
being $\Theta_m(\sigma)$. However, the isotropic subgaussian distribution $P = 0.5 \cN(0, 0.5\bI_d) + 0.5 \cN(0, 1.5\bI_d)$, 
which is not log-concave, does not satisfy this property---even for $m=2$---for any $\sigma =o(\sqrt{d})$.  
}

In summary, the aforementioned (sufficient) conditions 
for certifiable subgaussianity are fairly strong---numerous subgaussian distributions do not satisfy them.
That is, there is a significant gap in our understanding of the class 
of certifiably subgaussian distributions and how it compares 
with the class of all subgaussian distributions. 
This leads to the central question explored (and resolved) in this paper:
\begin{question} \label{ques:main}
{\em Can we characterize the class of certifiably subgaussian distributions ? } 
\end{question}
We believe that \Cref{ques:main} merits investigation in its own 
right and because of the range of potential implications 
on the computational complexity of several high-dimensional 
estimation tasks in robust statistics, clustering mixture models, 
and beyond.

\looseness=-1\Cref{ques:main} is not new. 
A very special case of this question, corresponding to certifiability of the fourth moments, 
was posed as one of the main open questions (Question 1.6) in~\cite{KotSteinhardt17}. Quoting \cite{KotSteinhardt17}:
\begin{quote}
``Is it the case that sum-of-squares certifies moment tensors 
for all sub-Gaussian distributions? Conversely, are there sub-Gaussian distributions that sum-of-squares
cannot certify? Even for 4th moments, this is unknown [...]''
\end{quote}
Moreover, \Cref{ques:main} has been explicitly highlighted in tutorial talks 
by experts on the topic; see, e.g.,~\cite{Li18-youtube,Hopkins18-ttic,Kothari20-youtube}. 
Interestingly, a common hypothesis within that community was that 
the class of certifiably subgaussians is a {\em proper} subset 
of subgaussians.\footnote{We note that the KLS conjecture is known to imply 
 that the class of logconcave distributions is certifiably bounded~\cite{KotSteinhardt17}.
 Note, however, that there exist subgaussian distributions that are not logconcave; 
 so even assuming the KLS conjecture does not imply that every subgaussian distribution is certifiably bounded.}
While this hypothesis turns out to be false (as shown in this paper), 
until our work it was plausible for a number of reasons; 
see discussion paragraph in the following subsection.

\subsection{Main Result: {\em All} Subgaussian Distributions are Certifiably So} \label{ssec:main-result}

As our main contribution, we answer \Cref{ques:main} 
by establishing the following theorem. 

\begin{restatable}[Certifiability of Subgaussian Distributions]{theorem}{ThmCertSub}
\label{thm:certifiability}
  There exists a universal constant $C > 0$ such that the following holds.
  Let $X \sim P$ be an $s$-subgaussian random vector over $\R^d$.  
  Then $P$ is $(Cs\sqrt{m}, m)$-certifiably bounded for any even $m$.
  In particular, $P$ is $Cs$-certifiably subgaussian.
\end{restatable}

Before we summarize the algorithmic applications of our main result, 
some remarks are in order. Interestingly, our proof of \Cref{thm:certifiability} 
establishes {\em the existence} of sum-of-squares 
certificates for arbitrary subgaussian distributions 
{\em without explicitly constructing the 
corresponding polynomials}. 
This is one of very few instances in the  SoS algorithmic statistics literature, 
where the underlying proofs are {\em implicit}.
(Another example is the recent work \cite{bakshi2024efficient}.)

At a high level, our approach is as follows: 
we use duality to reformulate our goal as that of analyzing the supremum 
of a certain \emph{non-linear} empirical process.
This reformulation allows us to leverage a deep result of Talagrand in probability theory, 
stating that the expected behavior of any \emph{linear} subgaussian process 
can be controlled by the analogous Gaussian behavior~\cite{Talagrand21}. 
Our main technical insight is to reduce the setting of the aforementioned 
nonlinear process to a linear process, enabling us to use Talagrand's result.  
We believe that our proof approach may yield analogous SoS-proofs 
for broader classes of distributions (see \Cref{sec:concl} for a discussion). 

\looseness=-1
\paragraph{Discussion}
\Cref{thm:certifiability} is surprising for a few reasons.
First, an efficient reduction of~\cite{BarBHKSZ12}, 
adapted to the present setting by Hopkins and Li~\cite{HopLi19}, 
implies the following: under the Small Set Expansion Hypothesis (SSEH)~\cite{RagSte10}, 
for any universal constants $m \in \N$ and $s>0$,  
there exists a multivariate distribution $P$ that is 
$(O(\sqrt{i}), 2i)$-bounded for all $i \in [m]$, while not 
being $(s, 4)$-certifiably bounded. That is, $P$ has 
``subgaussian moments'' up to order $2m$, while not even 
having certifiably bounded moments up to order $4$. This 
hardness result has led to the speculation by the 
community (see, e.g., \cite{Li18-youtube,Hopkins18-ttic})
that there exist subgaussian distributions that are not certifiably subgaussian. 
It is important to note that there is no contradiction between 
the latter result and \Cref{thm:certifiability}, 
because subgaussianity posits that moments of {\em all} orders are bounded.\footnote{It is not difficult 
to show that it suffices to assume bounds on moments of order up to $d \polylog(d)$; 
see the discussion in \Cref{sec:concl}.}

Additionally, some of the algorithmic implications of \Cref{thm:certifiability} 
on efficiently learning high-dimensional distributions 
stand in contrast with conventional wisdom. Concretely, \Cref{thm:certifiability} 
implies that the algorithmic task of clustering mixtures of {\em arbitrary subgaussian} 
distributions---under mean separation assumptions---is no harder computationally 
than the task of clustering separated Gaussian mixtures (see \Cref{ssec:mixtures}). 
In contrast, it has been a folklore conjecture that the subgaussian case 
ought to be harder at least in some regimes; 
see, e.g., the discussion in ~\cite{ChaSV17}.\footnote{For example, as already mentioned, 
robust mean estimation of subgaussian distributions is computationally harder 
than that of the Gaussian distribution (cf.\ \Cref{rem:tradeoffs-robust-mean}).}

\subsection{Algorithmic Implications of \Cref{thm:certifiability}}
\label{sec:implications}
Here we briefly state a number of algorithmic implications 
of \Cref{thm:certifiability} (see \Cref{table:summary}). 
We defer the full statements of these results to \Cref{sec:algo-implications}.

To formally state our results, we need the notion of certifiably 
\emph{hypercontractive} subgaussian distributions, defined below.   
\begin{definition}[Hypercontractive Subgaussian Distributions]
    \label{def:hypercontractive}
    We say that a distribution $Q$ over $\R^d$ with mean $\mu$ is $s$-hypercontractive-subgaussian if 
    for all even $m\in \N$ and vectors $v \in\R^d$ it holds 
    $\left(\E[\langle v, X- \mu\rangle^m]\right)^{1/m} \leq C s \sqrt{m} \left(\E[\langle v, X- \mu\rangle^2]\right)^{1/2}$.
\end{definition}
That is, the $m$-th moment in a direction $v$ scales with the standard deviation in direction $v$.
In contrast, \cref{def:sg} posits that the $m$-th moment obeys a uniform bound over all directions $v$. 
The distribution families satisfying \Cref{def:sg,def:hypercontractive} 
are incomparable.\footnote{For example, $\cN(0,\vec \Sigma)$ satisfies \Cref{def:sg} 
only when $\vec\Sigma \preceq O(1)\vec I$, whereas it satisfies \Cref{def:hypercontractive} for all $\vec\Sigma$.
On the other hand, any univariate Bernoulli distribution satisfies \Cref{def:sg} (with constant $s$), 
while it does not satisfy \Cref{def:hypercontractive} (with constant $s$) if the bias goes to $0$.}
Moreover, these two different families of distributions lead to distinct 
behaviors for the tasks of (robust) mean estimation under Euclidean norm and Mahalanobis norm.

We also define the following certifiable variant of \Cref{def:hypercontractive}:
\begin{definition}[Certifiably Hypercontractive Subgaussian Distributions]
    \label{def:cert-hypercontractive}
    We say that a distribution $Q$ over $\R^d$ with mean $\mu$
    and covariance $\vec \Sigma$ is $s$-certifiably-hypercontractive-subgaussian 
    if for all even $m \in \N$, the polynomial $q_m(v):= \left(C s \sqrt{m} \, \left(\E_{X \sim Q}[\langle v, X - \mu \rangle^2] \right) \right)^{m/2} -  \E_{X \sim Q}[\langle v,X - \mu\rangle^m] $ is a sum of square polynomials in $v$. 
\end{definition}
Hypercontractive subgaussianity and subgaussianity are related 
via a whitening operation using the matrix $\vec\Sigma^{\dagger/2}$;  
here, for a PSD matrix $A$, $A^{\dagger}$ denotes the pseudoinverse of $A$, 
$A^{\dagger/2}$ is the PSD square root of $A^{\dagger}$, 
and $\vec\Sigma$ is the covariance matrix of an underlying probability distribution.

\begin{fact}
A distribution $Q$ with covariance $\vec\Sigma$ is $s$-(certifiably)-hypercontractive-subgaussian if and only the distribution of $\vec \Sigma^{\dagger/2} X$ is $s$-(certifiably)-subgaussian.
\end{fact}

As a consequence of \Cref{thm:certifiability}, we obtain the following implication: 
\begin{theorem} \label{thm:certifiability-hyper}
    Let $P$ be an $s$-hypercontractive-subgaussian distribution over $\R^d$.
    Then $P$ is $Cs$-certifiably-hypercontractive-subgaussian over $\R^d$. 
\end{theorem}

We now state algorithmic implications of \Cref{thm:certifiability,thm:certifiability-hyper} in the table below, 
with details deferred to \Cref{sec:algo-implications}.  
For many of the estimation tasks below, the resulting algorithmic guarantees 
are qualitatively optimal, matching existing computational lower bounds; see \Cref{sec:algo-implications}.

\begin{table}[H]
\centering
\caption{We state some implications of \Cref{thm:certifiability} below.
For all of these applications, $\epsilon$ denotes the fraction of (arbitrary) outliers 
(except for clustering, where the implications are new even for $\eps=0$).
Here $m\in \N$ denotes an algorithmic parameter; the sample complexity $n$ 
of the resulting algorithm scales as $n = \poly(d^m,1/\eps)$, and the runtime scales as $(nd)^{\poly(m)}$. 
By ``No general algorithm'' below, we mean that no prior algorithm was known to 
achieve  error that did not scale (polynomially) with the condition number 
of the covariance matrix.\protect\footnotemark     ~For brevity, we hide the dependence 
on other problem-specific parameters, restriction on parameter regimes, 
and  polynomial dependence on $m$ in the final error guarantees.
}
\label{table:summary}
\begin{tabular}{@{}lcccc@{}}
\toprule
Estimation Task                                                                        & \begin{tabular}[c]{@{}c@{}}Inlier\\Distribution\end{tabular} & \begin{tabular}[c]{@{}c@{}}Information-\\
theoretic\\ Error\end{tabular}  & \begin{tabular}[c]{@{}c@{}}Previous Best\\ Guarantee in \\ Polynomial Time\end{tabular} 
& \begin{tabular}[c]{@{}c@{}} New Guarantees \\ from \Cref{thm:certifiability}\end{tabular} \\ \midrule
\addlinespace[0.7em]
\begin{tabular}[c]{@{}l@{}}Mean estimation: \\ Euclidean norm\end{tabular}               & subgaussian                    & $\widetilde{\Theta}(\eps)$                                   & 
            
           \begin{tabular}[c]{@{}c@{}}$\sqrt{\eps}$ \\ \scriptsize\cite{DiaKKLMS17,SteCV18}
           \end{tabular}
                                                                                  &
                         \begin{tabular}[c]{@{}c@{}}
$\eps^{1-1/m}$  \\ \scriptsize\cite{HopLi18,KotSS18}\\ \scriptsize (\Cref{cor:robust-euclidean})\end{tabular}                                  \\ 
\addlinespace[0.7em]
\begin{tabular}[c]{@{}l@{}}List-decodable \\ Mean estimation\end{tabular}          & subgaussian                 & $\widetilde{\Theta}(\eps)$                                     & 
\begin{tabular}[c]{@{}c@{}}$\sqrt{\frac{1}{1-\eps}}$ \\ \,\,\,\,\scriptsize \cite{ChaSV17} \end{tabular}
                                                                        & \begin{tabular}[c]{@{}c@{}}
$ \left(\frac{1}{1-\eps}\right)^{-\Omega(\frac{1}{m})}$  \\ \scriptsize\cite{KotSS18}\\ \scriptsize (\Cref{cor:list-decodable-euclidean})\end{tabular}                                            \\ 
\addlinespace[0.7em]
\begin{tabular}[l]{@{}l@{}}Mixture models:  \\
Mixture of $k$\\
$\Delta$-separated \\
components
\end{tabular} &   \begin{tabular}[c]{@{}l@{}}Each component \\is   subgaussian                                                     \end{tabular}  
& 
$\Delta\gtrsim \sqrt{\log k}$ &                                     \begin{tabular}[c]{@{}c@{}}$\Delta \gtrsim k^{\Omega(1)}$ \\ \,\,\,\,\cite{AchMcs05} \end{tabular}
                                                 & \begin{tabular}[c]{@{}c@{}}
$\Delta \gtrsim k^{O\left(\frac{1}{m}\right)}$  \\ \scriptsize\cite{HopLi18,KotSS18}\\ \scriptsize (\Cref{thm:clustering})\end{tabular}                                                            \\ 
\addlinespace[1em]
\begin{tabular}[c]{@{}l@{}}Mean estimation: \\ Mahalanobis norm\end{tabular}             & \begin{tabular}[c]{@{}c@{}}hypercontractive \\ subgaussian\end{tabular}                                          & $\widetilde{\Theta}(\eps)$               & \begin{tabular}[c]{@{}c@{}} No general \\  algorithm\end{tabular}                                                                        & \begin{tabular}[c]{@{}c@{}}
$\eps^{1-1/m}$  \\ \scriptsize\cite{HopLi18,KotSS18}\\ \scriptsize (\Cref{cor:cov-estimation})\end{tabular}                                      \\ 
\addlinespace[0.7em]
\begin{tabular}[c]{@{}l@{}}Covariance estimation: \\ Relative spectral \\
norm\end{tabular}          & \begin{tabular}[c]{@{}c@{}}hypercontractive \\ subgaussian\end{tabular}      & $\widetilde{\Theta}(\eps)$                                                  & \begin{tabular}[c]{@{}c@{}} No general \\  algorithm\end{tabular}                                                                        & \begin{tabular}[c]{@{}c@{}}
$ \eps^{1-2/m}$  \\ \scriptsize\cite{KotSS18}\\ \scriptsize (\Cref{cor:cov-estimation})\end{tabular}                                            \\
\addlinespace[0.7em]
\begin{tabular}{@{}l@{}}Linear regression:\\ Arbitrary noise\end{tabular} &           \begin{tabular}[c]{@{}c@{}}hypercontractive \\ subgaussian\end{tabular}                                        & $\widetilde{\Theta}(\eps)$                  &
\begin{tabular}[c]{@{}c@{}} No general \\  algorithm\end{tabular}
& \begin{tabular}[c]{@{}c@{}}
$\eps^{1 - \frac{2}{m}}$  \\ \scriptsize\cite{BakPra21}\\ \scriptsize (\Cref{cor:rob-regression})\end{tabular}                                                            \\ \bottomrule
\end{tabular}
\end{table}
\footnotetext{Recall that the condition number of the (unknown) covariance matrix of the underlying \emph{inliers} could be arbitrarily large, which renders those guarantees vacuous in the general setting.}

\bigskip

\label{fact:sseh-moments-hardness}

\subsection{Overview of Techniques} \label{ssec:techniques}

\input{techniques.tex}

\subsection{Related and Prior Work} \label{ssec:related}
\paragraph{Certifiable Boundedness of Probability Distributions} 
Certifiable boundedness and subgaussianity have extensive applications 
in algorithmic statistics, as discussed in \Cref{sec:implications}.
For a detailed treatment of algorithmic robust statistics, 
the reader is referred to~\cite{DiaKan22-book} .
Until our work, only certain structured distributions were known to have certifiably-bounded moments.
Distributions that either have independent coordinates 
or are rotationally invariant were shown to be certifiably bounded in \cite{HopLi18,KotSS18}.
Furthermore, \cite{KotSS18} proved that \emph{Poincare} distributions 
also have certifiable bounded moments (with the certifiable bound scaling with the Poincare constant).
The class of certifiable distributions can be further expanded by observing 
that the certifiable bounded moments property is (approximately) 
preserved under various natural operations, such as: (i) affine transformations, 
(ii) convolutions,  (iii) taking a weighted mixture of distributions whose means 
are close to each other, and (iv) small perturbations of moment-tensors.\footnote{The last of these 
in particular implies that the uniform distribution over a set of i.i.d.\ 
samples from a certifiable distribution is, with high probability, also certifiable.
}
The distributions obtainable by such operations were, before our work, the only ones known to have certifiably-bounded moments.

Recent work \cite{bakshi2024efficient} studied sum of squares proofs 
to certify \emph{anticoncentration}, a distinct property of some high-dimensional 
probability distributions which has been extensively useful in algorithmic statistics.

\paragraph{Injective Norms}
Our work shows that strong sum-of-squares upper bounds 
exist which bound from above the injective norms of a certain 
family of $m$-tensors---those arising as the $m$-th moments of a subgaussian distribution.
Sum-of-squares upper bounds on injective tensor norms 
(equivalently, upper bounds on homogeneous polynomials over the unit sphere) 
have been studied extensively for several classes of tensors (equivalently, polynomials).
\cite{doherty2012convergence,bhattiprolu2017weak} study worst-case $m$-tensors.
\cite{hopkins2015tensor,raghavendra2017strongly,potechin2020machinery,BarBHKSZ12,ge2015decomposing,guruswami2024certifying} study random $m$-tensors drawn from various probability distributions.

Under the exponential time hypothesis, approximating the injective norm 
of a tensor to any constant factor in polynomial time is not possible; 
this holds even for tensors which arise as the $m$-th moments 
of a probability distribution \cite{BarBHKSZ12}.
Under the small-set expansion hypothesis, constant-factor approximating 
injective tensor norm for $m$-th moment tensors is hard in polynomial time even 
for the moment tensors of distributions which have $Cm$ subgaussian moments, 
for any constants $C > 0$ and $m\geq 4$ \cite{BarBHKSZ12,HopLi19}.
This shows that (assuming the small-set expansion hypothesis) our result 
cannot be strengthened to apply to the $m$-th moment tensors 
of all distributions with a constant number of subgaussian moments.

The injective tensor norm of a tensor of the form $\sum_{i=1}^n a_i^{\otimes m}$
is equivalently the $2$-to-$m$ norm of the matrix $\bA$ whose rows 
are the vectors $a_1,\ldots,a_n$.
Approximation algorithms and hardness results for such ``hypercontractive norms'' 
have been studied recently in \cite{bhattiprolu2023inapproximability,bhattiprolu2019approximability}.
Our result shows that if $\bA$ is a random matrix with subgaussian rows, 
for $n$ sufficiently large, there is with high probability 
a sum of squares proof which bounds from above the $2$-to-$m$ norm of $\bA$.

\paragraph{Generic Chaining and Majorizing Measures} Generic chaining has 
seen numerous applications in theoretical computer science, 
especially in dimensionality reduction---see \cite{nelson2016chaining}.

\section{Preliminaries}

We now introduce the notation and state the background required to establish our main result.

\paragraph{Notation}
For $n\in\N$, we shall use $[n]$ to denote the set $\{1,\dots,n\}$.
For a random variable $Y$ over $\cY$ (which for us will shall be $\R^d$),
a function $h:\cY \to \R$, and a natural number $q \in \N$, we define the $L_q(Y)$-norm of $h$ as follows:
\begin{align}
    \|h(Y)\|_{L_q(Y)} := \left(\E_{Y}[\left|h(Y)\right|^q]\right)^{\frac{1}{q}}.
\end{align}
When the random variable $Y$ is clear from the context, 
we simply write $\|h(Y)\|_{L_q}$.
For a collection of random variables $Y_1,\dots,Y_n$, we use $Y_{1:n}$ to denote the concatenated vector $(Y_1,\dots,Y_n)$.
For distributions $P_1,\dots,P_m$, we denote the resulting product distribution by $\prod_{i=1}^mP_i$.
For the special case when all $P_i=P$, the product distribution is denoted by $P^{\otimes m}$.

For a vector $x$, we use $\|x\|_2$ to denote its Euclidean norm.
We use $\lesssim, \gtrsim$ to hide positive constants, that is, $a \lesssim b$ means that $a \leq C b$ for an absolute constant $C$ independent of the problem parameters.
Throughout this paper, the maximum over the empty set is set to be zero; this is because we shall be taking maximum over non-negative numbers.

\subsection{Subgaussian Distributions}

We will use the following simple claim that a concatenated vector of independent subgaussian samples is also subgaussian.
    \begin{claim}[Product of subgaussians is subgaussian]
    \label{claim:product}
        Let $Y_1,\dots,Y_n$ be $n$ independent samples from $P_1,\dots,P_n$, each over $\R^d$, with $P_i$ being a zero-mean $s_i$-subgaussian distribution.
        Let $Y' = Y_{1:n}$ be the concatenated vector in $\R^{nd}$.
        Then $Y'$ is zero mean and $O(\max_i s_i)$-subgaussian.
    \end{claim}
 \begin{proof}
We begin by stating an alternative characterization of subgaussian distributions.
\begin{fact}[Alternate definition of subgaussianity; see, for example, {\cite[Section 2.5]{Vershynin18}}]
\label{fact:def-subgaussianity}    
A distribution $P$ over $\R^d$ with mean $\mu$ is $\Theta(s)$-subgaussian, 
for $s\geq0$, if and only if for all  vectors $v \in \R^{d}$ and $\lambda>0$, it holds that $\E[\exp(\lambda\langle v, X\rangle)] \leq 2\exp\left(s^2 \lambda^2\|v\|_2^2\right)$.  
\end{fact}

\noindent  Fix a  vector $v' \in \R^{nd} $ to be $v' = (v_1,\dots,v_n)$, 
           where $v_i \in \R^d$. Using \Cref{fact:def-subgaussianity} and independence of $Y_i$'s,
        we see that
    $\E[\exp(v^\top Y')] = \E[\exp(\sum_i v_i^\top Y_i)]
    = \prod_i \E[\exp(v_i^\top Y_i)]
    \leq \prod \exp(s_i^2\|v_i\|_2^2) = \exp(\sum_i s_i^2\|v_i\|_2^2) = \exp\left(\frac{\sum_i s_i^2\|v_i\|_2^2}{\|v\|_2^2} \|v\|_2^2\right)\leq \exp(\|v\|_2^2 \max_i s_i^2)$, where we use that $\|v\|_2^2 = \sum_i \|v_i\|_2^2$.
    The claim follows by  another application of \Cref{fact:def-subgaussianity}.
    \end{proof}

\subsection{Generic Chaining and Majorizing Measures}
The following result is the well-known majorizing measures result of Talagrand. 
\begin{fact}[\cite{Talagrand21}]
\label{fact:canonical-subgaussian}
Let $T \subset \R^m$ be a fixed set.
Let $G \sim \cN(0,\bI_m)$ and let $Y$ be a zero-mean $s$-subgaussian vector over $\R^m$.
Then
\begin{align}
    \E\left[\sup_{t \in T} \langle t, Y\rangle \right]  
\lesssim s   \,\cdot   \E\left[\sup_{t \in T} \langle t, G\rangle \right]
\end{align}  
\end{fact}
This is a deep result in probability theory, stating that the maximum of a canonical subgaussian process is upper bounded by that of corresponding Gaussian process.
The version stated above is from {\cite[Corollary 8.6.3]{Vershynin18}}.
As mentioned in \Cref{ssec:techniques}, we shall need stronger concentration properties of Gaussian processes to replace the $m^m$ factor with $m^{m/2}$, 
which we record below.
\begin{fact}
\label{fact:canonical-subgaussian-moment}
    In the setting of \Cref{fact:canonical-subgaussian},  the following result holds for the higher moments of the supremum $\sup_{t \in T} \langle t, Y\rangle$:
for all $q\geq 1$, it holds that
\begin{align}
\left\|\sup_{t \in T} \langle t, Y\rangle\right\|_{L_q}
\lesssim s   \,\cdot   \E\left[\sup_{t \in T} \langle t, G\rangle \right] + s \sqrt{q} \cdot \sup_{t \in T} \|t\|_2
\end{align}
\end{fact}
Observe that the $q$-dependent term---which dictates the tail growth of the supremum---depends only on the diameter of $T$ ($\max_{t \in T}\|t\|_2$), which is usually  
much smaller than the $\E\left[\max_{t \in T} \langle t, G\rangle \right]$.
This decomposition of the higher moment into two terms encodes high probability concentration around the expectation of the supremum.\footnote{It is also possible to derive our desired inequalities by integrating the tail probability, but we choose the version above for simplicity.}
This moment bound can be inferred from \cite[Theorem 8.5.5]{Vershynin18} and \cite[Theorem 8.6.1]{Vershynin18}; see also \cite[Exercise 8.6.5]{Vershynin18}.

\subsection{Sum-of-Squares: Proofs and Duality}
\label{sec:sum-of-squares-background}
We refer the reader to \cite{BarSte16-sos-notes,FleKP19-sos} for further background on sum-of-squares proofs.
In this section, we state the background necessary for our purposes.
\begin{definition}[Sum-of-Squares Proofs]
    	Let $v_1,\ldots,v_d$ be indeterminates and let $\cA$ be a set of polynomial inequalities $\{ p_1(v) \geq 0,\ldots,p_m(v) \geq 0 \}$.
	For a polynomial $r$, a Sum-of-Squares (SoS) proof of the inequality $r(v) \geq 0$ from axioms $\cA$ is a set of polynomials $\{r_S(v)\}_{S \subseteq [m]}$ such that each $r_S$ is a sum of square polynomials and $r(x) = \sum_{S \subseteq [m]} r_S(v) \prod_{i \in S} p_i(v).$
	If the polynomials $r_S(v) \cdot \prod_{i \in S} p_i(v)$ have degree at most $t$ for all $S \subseteq[m]$, we say that this proof is of degree $t$ and denote it by $\cA \sststile{t}{v} r(v) \geq 0$.  
	When $\cA$ is empty, we simply write $\sststile{t}{v} r(t) \geq 0$.

\end{definition}
Our analysis crucially relies on the following objects called pseudoexpectations, which are dual objects to SoS proofs (as detailed shortly).
\begin{definition}[Pseudoexpectations]
    Let $v_1,\ldots,v_d$ be indeterminates.
	A degree-$t$ pseudoexpectation $\pE$ over variables $v$ is a linear map $\pE \, : \, \R[v_1,\ldots,v_d]_{\leq t} \rightarrow \R$
	from degree-$t$ polynomials to $\R$ such that $\pE \Brac{p(v)^2} \geq 0$ for any $p$ of degree at most $t/2$ and $\pE 1 = 1$.
	If $\cA = \{p_1(v) \geq 0, \ldots, p_m(v) \geq 0\}$ is a set of polynomial inequalities, we say that $\pE$ satisfies $\cA$ if for every $S \subseteq [m]$,
	the following holds: $\pE [s(v)^2 \prod_{i \in S} p_i(v)] \geq 0$  for all polynomials $s(\cdot)$ such that $s(v)^2 \prod_{i \in S} p_i(v)$ has degree at most $t$. 
    When we need to emphasize that the pseudoexpectation $\pE$ is over variable $v$, we denote it by $\pE_v$.
\end{definition}

The following fact establishes the duality between SoS proofs and Pseudoexpectations.
\begin{fact}[Duality Between SoS Proofs and Pseudoexpectations: Constraints with Archimedian Property~\cite{JosHen16-archimedian}]
\label{fact:sos-duality}
We say that a system of polynomial inequalities $\cA$ over $x$ is Archimedean if for some real $M > 0$ it contains the inequality $\|x\|^2 \leq M$.
For every Archimedian system $\cA$ and every polynomial $p$ and every degree $m$, exactly one of the following holds.
\begin{enumerate}[leftmargin=*]
\item For every $\eps > 0$, there is an SoS proof of  $\cA \sststile{m}{x} p(x) + \eps \geq 0$.
\item There is a degree-$m$ pseudoexpectation satisfying $\cA$ but $\pE [p(x)] < 0$.
\end{enumerate}
\end{fact}
The Archimedian property is necessary in some settings for strong duality when  $\cA$ is allowed to be nonempty. 
The next result does not require Archimedian property on $\cA$ if it
is non-empty.\footnote{The result in \cite[Theorem 6.1]{Laurent2009-sos} is applicable since $K = \R^d$, which has non-empty interior. 
}
\begin{fact}[{Duality Between SoS Proofs and Pseudoexpectations: The Unconstrained Setting~\cite[Theorem 6.1]{Laurent2009-sos}}]
\label{fact:sos-duality-unconstrained}
Let $p(\cdot)$ be a polynomial of degree at most $m$ for some even $m \in \N$. 
Then exactly one of the following holds: 
(i)  $\sststile{m}{x} p(x) \geq 0$, 
(ii) there is a degree-$m$ pseudoexpectation $\pE$ over $x$ with $\pE [p(x)] < 0$.
\end{fact}

Beyond the duality above, the only non-trivial property of the pseudoexpectation operator that we shall use is the following
inequality:
\begin{fact}[Pseudoexpectations H\"older; for $t=3$, see {\cite[Corollary A.5]{BarBHKSZ12}}]
Let $m$ be even and let $\pE$ be a degree-$m$ pseudoexpectation over the variables $(g_i)_{i=1}^n$ and $(h_i)_{i=1}^n$. Then
\label{fact:peholder}
\begin{align*}
\pE\left[ \frac{1}{n} \sum_i g_i^{m-1} h_i \right] \lesssim \left( \E\left[ \frac{1}{n} \sum_i g_i^m \right]\right)^{1  -\frac{1}{m}}     \left(\pE\left[\frac{1}{n} \sum_i h_i^m\right]\right)^{1/m} 
\end{align*}

\end{fact}
\begin{proof}
    We follow the proof strategy of \cite{BarBHKSZ12} with the fact that the AM-GM inequality has an SoS proof; see, for example, ~\cite{BarKS15} and references therein for even older proofs.
In particular, for any even $m$, there exists a degree-$m$ SoS proof in the variables $w$ of the inequality 
\begin{align}
    \label{eq:sos-am-gm}
    \sststile{m}{w_1,\dots,w_m} \qquad \prod_{i=1}^m w_i  \leq \frac{\sum_{i=1}^m w_i^m}{m} 
\end{align}
Fix any feasible pseudoexpectation over the variables $g$ and $h$.
Consider the variables $(\overline{g})_{i=1}^n,(\overline{h})_{i=1}^n$ defined as follows: for all $i\in[n]$, $\overline{g}_i := \frac{g_i}{\left(\pE[\sum_{i=1}^ng_i^m]\right)^{\frac{1}{m}}}$  and $\overline{h}_i = \frac{h_i}{\left(\pE[\sum_{i=1}^nh_i^m]\right)^{\frac{1}{m}}}$.
    Applying \Cref{eq:sos-am-gm} to the variables $\overline{g}_i$ and $\overline{h}_i$ for all $i\in[n]$ separately and taking an average,
    we see that 
    \begin{align*}
    \sststile{m}{\{\overline{g}_{i=1}^n,\overline{h}_{i=1}^n\} } \qquad \frac{1}{n} \sum_{i=1}^n \overline{g}_i^{m-1} \overline{h}_i   \leq \frac{1}{n} \sum_{i=1}^n \left( \frac{(m-1)\overline{g}_i^{m}}{m} + \frac{\overline{h}_i^m}{m} \right)       \,.  
    \end{align*}
    The desired result follows by applying pseudoexpectation $\pE$ to the both sides and noticing that the right hand side is equal to $1$.   
\end{proof}
In fact, as the following claims shows, the above bound is tight in the worst case.

\begin{fact}[Variational Characterization of Certifiable Moments]
\label{fact:lpnorm}
    Let $x_1,\dots,x_n $ be arbitrary vectors in $\R^d$ and $m$ be even.
    Let $\pE$ be an arbitrary degree-$m$ pseduoexpectation over $v:=(v_i)_{i=1}^d$ with $\pE[p(v)] > 0$ for a polynomial $p$ of degree at most $m$.
    Let $\cC$ be the set of all degree-$m$ pseudoexpectations $\pE'$ over $u:= (u_i)_{i=1}^n$ and $v$ with $\pE'[p(v)] > 0$ and $\pE'[\sum_i u_i^m]>0$.
    Then
    \begin{align*}
       \frac{\pE\left[\tfrac{1}{n}\sum_{i=1}^n \langle v,x_i\rangle^m \right]}{\pE[p(v)]} &\leq \left( \sup_{ \substack{\pE' \in \cC}} \pE'\left[\frac{\tfrac{1}{n}\sum_{i=1}^n u_i^{m-1}\langle v,x_i\rangle  }{ \left(\pE'[p(v)]\right)^{\frac{1}{m}} \left(\pE'\left[\tfrac{1}{n} \sum_{i=1}^n u_i^m\right]\right)^{1 - \tfrac{1}{m}} } \right]\right)^m\,\,.
    \end{align*}
\end{fact}
\begin{proof}
First, consider the setting that $\pE[\langle v,x_i\rangle^m] > 0$ for some $i \in [n]$; otherwise, the left hand side is zero and the right hand side is non-negative since $m$ is even. 
Consider the pseudoexpectation $\pE'$ over $v$ and $u$,
which is defined as follows:
for any polynomial $p$ over $u$ and $v$,
set $\pE'[p(u_1,\dots,u_n,v_1,\dots,v_d)]:= \pE[p(v^\top x_1,\dots,v^\top x_n, v_1,\dots,v_d)]$, i.e., we set $u_i$ to be $v^\top x_i$ for all $i\in[n]$.
It can be seen that $\pE'$ is a valid pseudoexpectation over $v$ and $u$, belonging to $\cC$.
For this choice, the right hand is exactly equal to the left hand side.
\end{proof}

\begin{remark}[Unit Norm Constraint in the Axiom]
    \label{rem:unit-norm}
    Prior works~\cite{HopLi18,KotSS18} focus on a stronger proof system 
    by incorporating $\|v\|_2=1$ into the axiom. To be precise, these works 
    define a distribution $P$ over $\R^d$ with mean $\mu \in \R^d$ 
    to be $(m,B_m)$-certifiably-bounded if 
    $\{\|v\|_2^2=1\} \sststile{m'}{v}B_m^m - \E_{X \sim P}[\langle v, X - \mu\rangle^m ] \geq 0$, 
    for some $m' = O(m)$.
    Instead of the degree of the SoS proof being restricted to $m'=O(m)$, 
    one could allow SoS proofs of larger degrees 
    at the expense of higher running times 
    in the resulting algorithms (usually scaling as $(dn)^{\poly(m')}$).
    In the presence of the $\|v\|_2^2=1$ constraint, 
    allowing $m'$ to be larger than $m$ might help 
    (as opposed to the unconstrained proof system, 
    where degree-$m$ proofs are necessary and sufficient).     
    Since we work with a weaker proof system (without the unit norm constraint), \Cref{thm:certifiability,thm:certifiability-hyper} 
    apply to the definitions in  \cite{HopLi18,KotSS18} 
    and the resulting algorithmic applications.
\end{remark}

\section{Meta Comparison Theorem: Inheriting SoS Certificates} 
\label{sec:meta-comp}

Our main technical workhorse is the following comparison result about pseudoexpectations between subgaussian distributions and Gaussians,
which is a consequence of \Cref{fact:canonical-subgaussian,fact:canonical-subgaussian-moment}.
While the statement is notational heavy, it simply states that the worst-case pseuedopextation of any polynomial that is linear in the subgaussian data is upper bounded by analogous Gaussian data (even after normalizing with (pseudo)-expected value of other polynomials).
In the first reading, we suggest skipping the statements about the moment bounds $\|\cdot\|_{L_q}$ below.

\begin{theorem}[Comparison Theorem for Pseudoexpectations of Linear Functions]
\label{thm:comparison}
Let $P_1,\dots,P_n$ be independent zero-mean $s$-subgaussian distributions on $\R^d$.\footnote{More generally, let $(Y_1,\dots,Y_n)$ be jointly $s$-subgaussian.}
Define $Q$ to be $\cN(0,\bI_d)$ over $\R^d$.

Let $u$ and $v$ be vector-valued variables of size $N$ and $d$, respectively.
Fix an even $t\in \N$ and $a_1,\dots,a_{n'} \in \R$. 
Let $h_1(\cdot),\dots,h_{n'}(\cdot)$ be $n'$ polynomials in $u$ and $v$ with degrees at most $t$, and 
let $p_1(\cdot),\dots,p_n(\cdot)$ be fixed polynomials over $u$ and $v$ with degrees less than $t-1$.

Let $\cT$ be any set of degree-$t$ pseudoexpectations over $u, v$.\footnote{We emphasize that the set $\cT$ is defined before looking at the samples $Y_{1:n}$ below.}
Then
\begin{align}
\label{eq:comparison-1}
    \E_{Y_{1:n} \sim \prod_{i=1}^nP_i}&
    \left[ \sup_{ \substack {\pE \in \cT: \\
    \pE[h_j(u,v)] \neq 0 \,\, \forall j\in[n']}} \pE_{u, v,w}\left[ \frac{\tfrac{1}{n}\sum_{i=1}^n p_i(u,v) \langle Y_i, v\rangle}{\prod_{j=1}^{n'}\left(\pE\left[h_j(u,v)\right]\right)^{a_j}} \right] \right]\\
&\lesssim
    s \cdot \E_{G_{1:n} \sim Q^{\otimes n}}
    \left[ \sup_{ \substack {\pE \in \cT: \\
    \pE[h_j(u,v)] \neq 0 \,\, \forall j\in[n']}} \pE_{u, v,w}\left[ \frac{\tfrac{1}{n}\sum_{i=1}^n p_i(u,v) \langle G_i, v\rangle}{\prod_{j=1}^{n'}\left(\pE\left[h_j(u,v)\right]\right)^{a_j}} \right] \right] \;. \nonumber 
\end{align}
Moreover, if $Z$ denotes the random variable inside the expectation on the left hand side above and $R$ denotes the number on the right hand side, then
$Z$ satisfies the following moment bound  $\forall q \in \N$:
\begin{align}
\label{eq:comparison-2}
    &\left\|Z \right\|_{L_q(Y_{1:n})}\lesssim
    R +  s \sqrt{q} \cdot  \sup_{ \substack {\pE \in \cT: \\
    \pE[h_j(u,v)] \neq 0 \forall j\in[n']}}  \frac{\left( \sum_{i=1}^n \sum_{j=1}^d \left(\pE\left[\frac{p_i(u,v) v_j}{n}\right]\right)^2 
    \right)^{\tfrac{1}{2}}}{\prod_{j=1}^{n'}\left|\pE\left[h_j(u,v)\right]\right|^{a_j}}  \,, 
\end{align}
where $Y_1,\dots,Y_n$ are again sampled from $\prod_{i=1}^nP_i$.
\end{theorem}
\begin{proof}

    For any $\pE \in \cT$ over $u$ and $v$,
    we denote $t_{\pE} \in \R^{nd}$ to be the concatentation of the vectors $\left(\tfrac{ \tfrac{1}{n} \pE[p_1(u,v) v]}{\prod_{j=1}^{n'}\left(\pE\left[h_j(u,v)\right]\right)^{a_j}},\dots, \tfrac{\tfrac{1}{n}\pE[p_n(u,v) v]}{\prod_{j=1}^{n'}\left(\pE\left[h_j(u,v)\right]\right)^{a_j}}\right)$, with each component being a $d$-dimensional vector.
    We set $T$ to be the collection of all such vectors obtained by pseudoexpectations $\pE$ in $\cT$ with $\pE[h(u,v)] \neq 0$.
    Importantly, the set $T$ does not depend on the randomness of $Y_i$'s; it is defined purely in terms of the pseudoexpectations $\cT$ and the polynomials $p_i$ and $h$, which themselves were defined without looking at the samples $Y_i$'s.
    
    If $Y' \in \R^{nd}$ denotes the concatenation of $(Y_1,\dots,Y_n)$ and $G' \in \R^{nd}$ denotes the concatenation of $(G_1,\dots,G_n)$,
    then it can be seen that 
    \begin{align}
    \langle Y', t_{\pE}\rangle &= \pE_{u, v,w}\left[ \frac{\tfrac{1}{n}\sum_{i=1}^n p_i(u,v) \langle Y_i, v\rangle}{\prod_{j=1}^{n'}\left(\pE\left[h_j(u,v)\right]\right)^{a_j}} \right]\,\qquad \text{and}\\
    \langle G', t_{\pE}\rangle &= \pE_{u, v,w}\left[ \frac{\tfrac{1}{n}\sum_{i=1}^n p_i(u,v) \langle G_i, v\rangle}{\prod_{j=1}^{n'}\left(\pE\left[h_j(u,v)\right]\right)^{a_j}} \right] \nonumber\,.
    \end{align}

    Thus, the desired inequality in \Cref{eq:comparison-1} is equivalent to
    \begin{align}
    \label{eq:rephrased-claim}
        \E_{Y' }\left[ \sup_{t \in T} \langle t, Y'\rangle \right]
        \lesssim        
        \E_{G' }\left[ \sup_{t \in T} \langle t,  G' \rangle \right]\,.
    \end{align}
    By \Cref{claim:product}, $Y'$ is an $O(s)$-subgaussian vector. Since the distribution of $G'$ is simply $\cN(0,\bI_{nd})$, the desired results \Cref{eq:rephrased-claim} and \Cref{eq:comparison-2} follow from \Cref{fact:canonical-subgaussian,fact:canonical-subgaussian-moment}.

\end{proof}
  
\section{Certifiable Subgaussianity and Resilience} %
\label{sec:resilience_to_certifiable_moments}
We explore the consequences of the meta comparison theorem (\Cref{thm:comparison}) in this section to establish \Cref{thm:certifiability}.
First, we introduce the notion of resilience and certifiable resilience in \Cref{sec:warmup}, 
which motivates the notion of certifiable generalized resilience in \Cref{sec:certi-gen-resilience}.
Next, we connect certifiable generalized resilience to certifiable empirical moments in \Cref{sec:cert-emp-moments}, and finally establish \Cref{thm:certifiability} in \Cref{sec:cert-dist-moments}.

\subsection{Warmup: Certifiable Resilience and Efficient Robust Mean Estimation}
\label{sec:warmup}
Resilience is a natural concept that arises in robust statistics~\cite{SteCV18}, which is related to  the stability conditions introduced in \cite{DiaKKLMS16-focs}.
In fact, the property of resilience information-theoretically suffices for robust mean estimation.
Since the bounded moments condition in \Cref{eq:our-tech-bdd-moments-samples} is non-linear (in the $Y_i$'s), we focus on the resilience property first, which happens to be linear.

\begin{definition}[Resilience]
\label{def:resilience-basic}
Let $S$ be a multiset $S = \{y_1,\dots,y_n\} \subset \R^d$, $\eps > 0$, and $\delta\geq 0$.
We say that $S$ satisfies $(\epsilon, \delta)$-resilience, or is $(\eps,\delta)$-resilient, if there exists $\mu \in \R^d$
such that for all subsets $S' \subset S$ with $|S'| \leq \eps |S|$ and for all unit vectors $v$, it holds that $\frac{1}{n}\sum_{x \in S'} \langle x -\mu, v\rangle  \leq \delta$. 
\end{definition}
For the purpose of obtaining 
an efficient algorithm, we need that the resilience can be certified (via a low-degree SoS proof). This leads to the following definition of certifiable resilience: 
\begin{definition}[Certifiable Resilience]
\label{def:resilience}
Let $S$ be a multiset $S = \{y_1,\dots,y_n\} \subset \R^d$, $\eps > 0$, $\delta\geq 0$, and $m \in \N$.
We say that $S$ is $(\eps,\delta,m)$-certifiably-resilient  if there exists a degree-$m$ SoS proof of
\begin{align}
    \label{eq:def-certifiable-resilience}
\tfrac{1}{n}\sum_i \langle v, w_i (x_i -  \mu) \rangle \leq \delta
\end{align}
 in the variables $(w_i)_{i=1}^n,(v_i)_{i=1}^d$ of degree $O(m)$ under the constraints: (i) $\sum v_i^2 = 1$, (ii) $\sum_{i=1}^n w_i \leq \eps n$, and (iii) for all $i \in [n]$, $w_i^2 = w_i$.   
\end{definition}

\looseness=-1Encouragingly, the polynomial inequality in the certifiable resilience condition \Cref{eq:def-certifiable-resilience} is linear in the empirical samples $Y_1,\dots,Y_n$, permitting the use of \Cref{thm:comparison}.
Leveraging the fact that a set of $n=\poly(d^m/\eps)$ i.i.d.\ Gaussian samples  is $(\eps,m,\sqrt{m}\eps^{1-1/m})$-certifiably-resilient, \Cref{thm:comparison} implies that a set of $\poly(d^m)$ i.i.d. samples from a subgaussian distribution is also $(\eps,\sqrt{m}\eps^{1 - \frac{1}{m}},m)$-{certifiably} resilient; See the next subsection for details. 
This structural result already suffices to overcome the $\sqrt{\epsilon}$ barrier for robust mean estimation  
of a generic subgaussian distribution; See also \Cref{rem:cert-res-all-samples}. 
Hence, while certifying $(\eps, o(\sqrt{\eps}))$-resilience is conjectured to be computationally hard for, say, $(C, 1000)$-bounded distributions~\cite{HopLi19},
the above discussion shows that subgaussian distributions are not the hard instances.

\subsection{Certifiable Generalized Resilience}
\label{sec:certi-gen-resilience}

With certifiable resilience (\Cref{def:resilience}) as motivation, we begin by certifying \emph{generalized} resilience, which is a stronger property.
In the first reading, we recommend the reader to skip the statement about the $L_q$ moment bounds in \Cref{eq:multi-scale-resilience-moment}.

\begin{theorem}
\label{thm:multi-scale-resilience}
Let $P$ be an $s$-subgaussian distribution over $\R^d$ with mean $\mu$ and let $n \geq \poly(d^m)$.

For a fixed $N \in \N$,  Let $u = (u_i)_{i=1}^N$ and $v=(v_i)_{i=1}^d$ be variables and let $t$ be even.
For a fixed $n\in \N$ and degree $m \in \N$, fix the polynomials $p_1,\dots,p_n$ over $u,v$ of degree at most $m'$. 
Let $\cT$ be any fixed set of degree $mm'$ pseudoexpectations over the variables $u,v$.%
Then  
\begin{align*}
\E_{Y_{1:n}}\left[
\max_{\substack{\pE \in \cT: \\
\pE[\|v\|_2^m]>0,\\
\pE\left[\sum_{i=1}^n p_i^m(u,v)\right] >0
}}
\frac{\pE\left[\tfrac{1}{n}\sum_{i=1}^n p_i^{m-1}(u,v) \langle Y_i ,v\rangle\right]}{\left(\pE[\|v\|_2^m]\right)^{1/m}\left( \pE \left[ \frac{1}{n}\sum_i p_i^{m}(u,v)\right] \right)^{1 - \frac{1}{m}}}
       \right]
       &\lesssim s \sqrt{m}\,\,\,.
\end{align*}
Moreover, if we define $Z$ to be the random variable inside the expectation above, 
then the following moment bound is true for all  $q \geq 1$: 
\begin{align}
\label{eq:multi-scale-resilience-moment}
\left\|Z       \right\|_{L_q(Y_{1:n})}
       \lesssim s \sqrt{m} +
       s \sqrt{q} \frac{\sqrt{d}}{n^{\frac{1}{m}}} \,.
\end{align}
\end{theorem}
To see that resilience (\Cref{def:resilience}) is a special case of 
\Cref{thm:multi-scale-resilience}, choose the polynomial $p_i(u,v):=u_i$ and enforce 
$\cT$ to satisfy the following constraints: (i) binary constraints $u_i^2=u_i$ for all $i$, 
(ii)  the cardinality constraint $\sum_i u_i^2 \leq \eps/n$ to $\cC$, 
and (iii) the unit vector constraint on $v$: $\sum v_i^2=1$.
For this special case, the result states that (on average), 
the pseudoexpectation of any $\pE[\tfrac{1}{n}\sum_i u_i \langle Y_i,v\rangle ]$ 
is at most $\sqrt{m}\eps^{1-1/m}$ given $\poly(d^m)$ i.i.d.\ samples from a subgaussian 
distribution, where we use that a set of $\poly(d^m)$ i.i.d.\ samples 
from a Gaussian distribution is $(\eps,\sqrt{m}\eps^{1-1/m},O(m))$-certifiably 
resilient for all $\eps$.

Some remarks are in order.
\begin{remark}[Certifiable Resilience at All Sample Sizes]
\label{rem:cert-res-all-samples}
In fact, the same  argument implies that for any arbitrary $n$ 
(not necessarily bigger than $\poly(d^m)$) and even $m$, 
the certifiable resilience parameter $\delta=\delta(\eps,m)$ of a set of n i.i.d.\ samples 
from a  subgaussian distribution is bounded above  
by that of Gaussians (up to a constant factor).  
\end{remark}

\begin{remark}[Certifiable Resilience Versus Robust Mean Estimation]
Since certifiable resilience suffices for computationally efficient robust mean estimation, 
combining \Cref{rem:cert-res-all-samples}, \Cref{rem:tradeoffs-robust-mean}, 
and the algorithmic result of \cite{DiaKKLMS16-focs} shows 
that certifying resilience of Gaussian data is computationally harder (under certain conjectures) 
than robust mean estimation of Gaussian inliers.\footnote{See also \cite{KotMZ22} 
who provided efficient robust mean estimation of Gaussians without certifying resilience.}
\end{remark}

We now provide the proof of \Cref{thm:multi-scale-resilience}.
\begin{proof}
By \Cref{thm:comparison},
it suffices to upper bound the expression for Gaussians (up to a factor of $s$).
Applying pseudoexpectation H\"older's inequality (\Cref{fact:peholder}) to the numerator,
we obtain that
\begin{align*}
       \E_{G_1,\dots,G_n \sim Q^{\otimes n}}&
    \left[
        \max_{\substack{\pE \in \cT: \\
    \pE[\|v\|_2^m]>0,\\
    \pE\left[\sum_{i=1}^n p_i^m(u,v)\right] >0
    }} 
        \pE_{u, v}\left[ \frac{\tfrac{1}{n}\sum_{i=1}^n p_i^{m-1}(u,v) \langle G_i, v\rangle}{
    \left(\pE[\|v\|_2^m]\right)^{1/m}
    \left( \pE \left[ \frac{1}{n}\sum_i p_i^{m}(u,v)\right] \right)^{1 - \frac{1}{m}}  }  \right]\right]\\ 
&\leq       \E_{G_1,\dots,G_n \sim Q^{\otimes n}}
        \left[ 
        \max_{\substack{\pE \in \cT: \\
    \pE[\|v\|_2^m]>0,\\
    \pE\left[\sum_{i=1}^n p_i^m(u,v)\right] >0
    }}
    \frac{\left( \pE\left[\frac{1}{n}\sum_i \langle G_i, v\rangle^m \right]  \right)^{\frac{1}{m}} \left( \pE\left[\frac{1}{n} \sum_i p_i^m(u,v)\right] \right)^{1 - \frac{1}{m}}}{\left(\pE[\|v\|_2^m]\right)^{1/m} \left( \pE \left[ \frac{1}{n}\sum_i p_i^{m}(u,v)\right] \right)^{1 - \frac{1}{m}}} \right]\\
    &=       \E_{G_1,\dots,G_n \sim Q^{\otimes n}}
    \left[\max_{\substack{\pE \in \cT: \\
        \pE[\|v\|_2^m]>0
        }}
    \frac{\left( \pE\left[\frac{1}{n}\sum_i \langle G_i, v\rangle^m \right]\right)^{\tfrac{1}{m}}}{\left(\pE[\|v\|_2^m]\right)^{1/m}}\right]\\
    &\lesssim \sqrt{m},
\end{align*}
where the last step uses that the empirical distribution over a large enough 
set of i.i.d.\ samples from the Gaussian distribution is certifiably bounded~\cite{KotSS18,HopLi18}.

We now control the \emph{weak deviation term}---the term that is multiplied to $s\sqrt{q}$--- in \Cref{eq:comparison-2} as follows:
for any $\pE \in \cT$, the corresponding term is the square root of the expression below
\begin{align*}
     \frac{\sum_{i=1}^n \sum_{j=1}^d \left(\pE\left[\frac{1}{n}p_i^{m-1}(u,v) v_j\right] 
    \right)^{2}}{ \left(\pE[\|v\|_2^m]\right)^{2/m}\left(\pE\left[ \frac{1}{n}\sum_i p_i^{m}(u,v)\right]\right)^{2\left(1 - \frac{1}{m}\right)} } 
&= \frac{1}{n^{\frac{2}{m}}}      \frac{\sum_{i=1}^n \sum_{j=1}^d \left(\pE\left[p_i^{m-1}(u,v) v_j\right] 
    \right)^{2}}{ \left(\pE[\|v\|_2^m]\right)^{2/m} \left(\pE\left[ \sum_i p_i^{m}(u,v)\right]\right)^{2\left(1 - \frac{1}{m}\right)} }  \\
&\leq  \frac{1}{n^{\frac{2}{m}}}      \frac{\sum_{i=1}^n \sum_{j=1}^d \left(\pE\left[p_i^{m}(u,v)\right] \right)^{2 - \frac{2}{m}} \left(\pE\left[v_j^m\right]\right)^{\frac{2}{m}}}{ \left(\pE[\|v\|_2^m]\right)^{2/m} \left(\pE\left[ \sum_i p_i^{m}(u,v)\right]\right)^{2\left(1 - \frac{1}{m}\right)} }  \tag*{(\Cref{fact:peholder})} \\
&=  \frac{1}{n^{\frac{2}{m}}} \frac{ \left(\sum_{j=1}^d \left(\pE\left[v_j^m\right]\right)^{\frac{2}{m}} \right)}{ \left(\pE[\|v\|_2^m]\right)^{2/m} }
\frac{\left(\sum_{i=1}^n \left(\pE\left[p_i^{m}(u,v)\right] \right)^{2 - \frac{2}{m}}\right)}{ \left(\pE\left[ \sum_i p_i^{m}(u,v)\right]\right)^{2\left(1 - \frac{1}{m}\right)}}
\,.    \\
\end{align*}
Defining $x \in \R^n$ to be $x_i := \pE[p_i^m(u,v)] \geq 0$ and $r := 2-\frac{2}{m} \geq 1$, the fraction corresponding to $p_i$'s above  is equal to $\tfrac{\sum_{i=1}^n \left(x_i \right)^{2 - \frac{2}{m}}}{ \left(\sum_i x_i\right)^{2\left(1 - \frac{1}{m}\right)} }     = 
     \frac{\|x\|_r^r}{\|x\|_1^r} \leq 1$, 
where we use that $\|x\|_1 \geq \|x\|_r$ for all $r \geq 1$.
Turning our attention to the fraction involving $v$, we
use the inequality that $\pE[v_j^m] \leq \pE[\|v\|_2^m]$.\footnote{This follows by the fact that $\|v\|_2^m - v_j^m = ( (\sum_{i \neq j} v_i^2) + v_j^2)^{m/2} -  (v_j^2)^{m/2}$ is a sum of squares polynomial in $v$.}
This implies that the numerator in the corresponding fraction is less than 
$\sum_{j=1}^d(\pE[\|v\|_2^m])^{2/m} = d(\pE[\|v\|_2^m])^{2/m}$. 
Therefore, the corresponding fraction is less than $d$. 
Hence, the weak deviation term is at most $\frac{\sqrt{d}}{n^{1/m}}$.
\end{proof}
\subsection{Certifiable Empirical Moments from Multi-scale Resilience}
\label{sec:cert-emp-moments}

We now use \Cref{thm:multi-scale-resilience} to establish the certifiablility of empirical moments of a subgaussian distribution.

\begin{theorem}
\label{cor:certifiable-moments}
Let $P$ be an $s$-subgaussian distribution over $\R^d$ with mean $\mu$ and let $n \geq \poly(d^m)$.
Suppose further that $n \geq \poly(d^{m})$.
Let $\cT'$ be the set of all degree-$m$ pseudoexpectations over $v=(v_i)_{i=1}^d$ such that $\pE[\|v\|_2^m] > 0$.
Then
$\E_{Y_1,\dots,Y_n \sim P^{\otimes n}}\left[\max_{\pE \in \cT'} \frac{\pE\left[\tfrac{1}{n}\sum_i \langle v,Y_i\rangle^m\right]}{\pE[\|v\|_2^m]} \right] \lesssim \left(O(s\sqrt{m})\right)^m$.
\end{theorem}
\begin{proof}
Let $\cT$ be the set of degree-$m$ pseudoexpectations over $(u_i)_{i=1}^n$ and $v$ with $\pE[\|v\|_2^m] > 0$.
Then \Cref{fact:lpnorm} implies the following deterministic inequality for any $\pE \in \cT'$:
\begin{align}
    \max_{\pE \in \cT'} \frac{\pE\left[\tfrac{1}{n}\sum_i \langle v,Y_i\rangle^m\right]}{\pE[\|v\|_2^m]}
    &\leq \left( \sup_{\substack{\pE \in \cT:\\
    \pE[\sum_iu_i^m] > 0}} \frac{\pE\left[\tfrac{1}{n}\sum_i u_i^{m-1} \langle v,Y_i\rangle\right]}{\left(\pE[\|v\|_2^m]\right)^{\frac{1}{m}} \left(\pE[\tfrac{1}{n}\sum_i u_i^{m}] \right)^{1 - \frac{1}{m}}}\right)^m\,.
\end{align}
Taking expectation over $Y_{1:n}$ on both sides and
applying \Cref{eq:multi-scale-resilience-moment} with $q = m$ from \Cref{thm:multi-scale-resilience} to upper bound the right hand side,
we get the desired conclusion.
 
\end{proof}

\subsection{Empirical Moments to Population Moments: Proof of \Cref{thm:certifiability}}
\label{sec:cert-dist-moments}
Finally, we complete the proof of \Cref{thm:certifiability} using \Cref{cor:certifiable-moments}.
\ThmCertSub*
\begin{proof}
Without loss of generality, we assume that $P$ is centered.

Let $\cT'$ be the set of $\Theta(m)$-pseudoexpectations with $\pE[\|v\|_2^m] >0$.
It suffices to restrict our attention to $\cT'$ because if $\pE[\|v\|_2^m] = 0$, then it is easy to see that $\pE[\langle x,v\rangle^m] \leq \|x\|_2^m \pE[\|v\|_2^m ]$ is also zero for any $x \in\R^d$.

To establish \Cref{thm:certifiability}, the duality between the pseudoexpectations and the SoS proofs (\Cref{fact:sos-duality-unconstrained}) implies that it suffices to show the following:
\begin{align*}
\max_{\pE \in \cT'} \pE\left[\frac{\E_{Y}\left[\langle v,Y\rangle^m\right]}{\pE[\|v\|_2^m]}\right] \lesssim \left(O(s\sqrt{m})\right)^m\,.
\end{align*}
Let $Y_1,\dots,Y_n$ be $n$ i.i.d.\ samples from $P$ for $n\gtrsim \poly(d^m)$.
    The following sequence of (in)equalities yield the desired claim:
    \begin{align*}
\max_{\pE \in \cT'} \pE&\left[\frac{\E_{Y}\left[\langle v,Y\rangle^m\right]}{\pE[\|v\|_2^m]}\right] \\
&= \max_{\pE \in \cT'} \pE\left[ \frac{{}\E_{Y_{1:n}}\left[\tfrac{1}{n}\sum_i \langle v,Y_i\rangle^m\right]}{\pE[\|v\|_2^m]}\right] \tag*{($Y, Y_1,\dots,Y_n$ are identically distributed)}\\
&= \max_{\pE \in \cT'}  \E_{Y_{1:n}} \left[ \frac{ \pE\left[ \tfrac{1}{n}\sum_i \langle v,Y_i\rangle^m\right]}{\pE[\|v\|_2^m]}\right] \tag*{(linearity of $\E$ and $\pE$)}\\
&\leq  \E_{Y_{1:n}}\left[\max_{\pE \in \cT'} \frac{\pE\left[ \tfrac{1}{n}\sum_i \langle v,Y_i\rangle^m\right]}{\pE[\|v\|_2^m]}\right] \tag*{(Jensen's inequality)}\\
&\leq  \left(O(s \sqrt{m})\right)^m,
    \end{align*}
    where the last step uses \Cref{cor:certifiable-moments}. 
    This completes the proof of \Cref{thm:certifiability}. 
\end{proof}

\input{applications}

\section{Conclusions and Open Problems} \label{sec:concl}

In this paper, we proved that all subgaussian distributions\footnote{We note that our result applies to a slightly larger distribution family: A simple truncation argument shows that \Cref{thm:certifiability} also applies to distributions that are $O(C\sqrt{m},m)$-certifiably bounded for $m \leq d \polylog(d)$ (as opposed to all $m \in \N$).} are certifiably subgaussian, i.e.,  that subgaussian moment bounds have a low-degree SoS proof. As a corollary, we obtained novel algorithmic applications for a number of well-studied high-dimensional learning tasks---in robust statistics and beyond. Our work suggests several 
intriguing research directions and concrete open problems.

\paragraph{Certifiability Beyond Subgaussians?} 
Perhaps the most obvious direction is to explore whether the statement of \Cref{thm:certifiability} holds for broader families of well-behaved distributions---with the class of {\em subexponential}\footnote{A distribution over $\R^d$ is $s$-subexponential if it is $(Csm,m)$-bounded for all $m \in \N$.}  distributions being the first natural candidate. At this time, it remains open whether subexponential distributions are even $(4, O(1))$-certifiably bounded. 
At a technical level, this is because \Cref{fact:canonical-subgaussian} 
is known to fail for subexponential distributions (for some $T$), 
even when comparing with the canonical distribution with subexponential tail, the Laplace distribution.
That said, we believe that the covering number of the set of pseudoexpectations 
that we establish in our work (for example, $\cT$ in \Cref{thm:multi-scale-resilience}) 
may be useful towards making further progress.\footnote{These covering numbers (in the Euclidean norm) 
follow from generic chaining guarantees~\cite{Talagrand21}. 
Analogous results for more general $\ell_p$ metrics 
follow from \cite[Chapter 8]{Talagrand21} by analyzing the corresponding product distributions.
}

\looseness=-1More broadly, it would interesting to close the gap between 
more general distribution families and their certifiable 
versions. Specifically, suppose that a distribution $P$ is $(m,B_m)$-
bounded. What is the largest $m'\geq 4$ and the smallest $B'_{m'}$ 
such that $P$ is $(m',B'_{m'})$-certifiably bounded 
(perhaps after allowing for degree-$m''$ SoS proofs 
for some $m''\geq m$ and under the axiom $\|v\|_2^2 =1$ as in \Cref{rem:unit-norm})?
The reduction-based hardness results of 
\cite{HopLi19} show that,  under SSEH,  for any constants $m, B_m$, 
even for the smallest non-trivial choice of $m' = 4$, the certifiable 
bound $B'_{4}$ must scale with the dimension $d$. 
To the best of our 
knowledge, the case of $m= \polylog(d), m'=4$ and 
$B_m = \poly(m), B'_4=\poly(m)$ remains open.  
We hope that our work will serve as motivation for further research on these important directions.

\paragraph{Faster Sample-Efficient Algorithms.}
Our main result yields efficient algorithms, 
via the Sum-of-Squares method, for a range of high-dimensional 
learning tasks involving subgaussian distributions (\Cref{sec:algo-implications}). The sample complexity of these algorithms
is qualitatively optimal (matching information-computation tradeoffs for SQ 
algorithms and low-degree tests). It would be interesting to obtain a more 
``fine-grained'' understanding of the complexity of these tasks, both with 
respect to sample and computational complexity. Can we 
design efficient learning algorithms that do not rely on 
solving large convex programs? For example, 
does there exist an algorithm for robust subgaussian mean estimation, 
achieving error $O(\eps^{3/4})$, using $n=O_\eps(d^2)$ samples 
and $\widetilde{O}(nd)$ runtime (as opposed to $\poly(nd)$ runtime)?

A related question is whether there exist non-trivial generalizations of gaussianity for which the aforementioned information-computation tradeoffs do not apply. For example, does there exist a natural subclass of subgaussian distributions for which we can robustly estimate the mean within error $\tilde{O}(\eps)$ in $\poly(d/\eps)$ time?
\

\section*{Acknowledgements}
We thank David Steurer for his $\ell_p$-norm reinterpretation of our proof.

\printbibliography

\end{document}

%% file: techniques.tex
We give an overview of the proof of \Cref{thm:certifiability}.
For the purpose of this intuitive explanation, we will establish 
the slightly weaker statement that the polynomial 
$(C m )^m \|v\|_2^m - \E_{X \sim P} \iprod{v,Y}^m$ is a sum of squares.
Substituting a stronger concentration bound for subgaussian processes 
at one key point in the proof leads to the improvement 
from $m^m$ to $m^{m/2}$---see the formal proof in \Cref{sec:resilience_to_certifiable_moments}.

\paragraph{Duality} Our starting point is the duality between 
sum of squares polynomials and \emph{pseudoexpectations}.
A degree-$m$ pseudoexpectation $\pE$ over a variable $v= (v_i)_{i=1}^d$ 
is an $\R$-valued linear operator on the vector space 
of degree-$m$ polynomials in $v$---that is, $\pE$ assigns a 
real number to every polynomial $p(v)$ with degree at most $m$.
Such a linear operator is a pseudoexpectation if it satisfies: 
$\pE[1] = 1$ (normalization) and $\pE[p(v)^2] \geq 0$ for every $p$ 
with degree at most $m/2$ (positivity).
We refer the reader to \Cref{sec:sum-of-squares-background} for further details.

Duality between pseudoexpectations and sum of squares polynomials 
states that a degree at most $m$ polynomial $p(v)$ is a sum of squares if and only if 
for all degree-$m$ pseudoexpectations, $\pE[p(v)] \geq 0$; 
see \Cref{fact:sos-duality,fact:sos-duality-unconstrained}.
In our context, duality means that it will be enough to show 
\[
\pE [ \E_{Y \sim v} \iprod{v,Y}^m] \leq (C m)^m \cdot \pE [\|v\|_2^m]
\]
for every pseudoexpectation $\pE$ of degree $m$, or, equivalently,
\[
\sup_{\pE \text{ of degree } m} \frac{ \pE [ \E_{Y \sim v} \iprod{v,Y}^m] }{ \pE [\|v\|_2^m]} \leq (C m)^m \, .
\]

\paragraph{A Nonlinear Empirical Process}
The first step of our proof is to replace $\E_{Y \sim v} \iprod{v,Y}^m$ 
with an empirical average over $n$ samples.
In fact, using Jensen's inequality and linearity of the operators, 
it suffices to show the following result (see \Cref{sec:cert-dist-moments} for details):
\begin{align}
\label{eq:our-tech-bdd-moments-samples}
 \E_{Y_1,\dots,Y_n \stackrel{\text{i.i.d.}}{\sim} P}\left[\sup_{\pE}  \frac{\pE\left[\frac{1}{n}\sum_{i=1}^n\langle v, Y_i\rangle^m \right]}{\pE \|v\|_2^m}  \right] \leq (C m)^m \, ,
\end{align}
where $Y_1,\dots,Y_n$ are i.i.d.\ samples from $P$, 
for some sufficiently large $n$.\footnote{In fact, using uniform convergence arguments, 
it can be seen that this step is not loose for $n$ large enough.}
Importantly, we note that the expectation over the samples is now outside the supremum. 

Thus, duality allows us to reduce our task to analyzing the (expected) supremum
of the empirical process in \eqref{eq:our-tech-bdd-moments-samples}.
An empirical process, associated to a function class $\cF$ and a distribution $P$, 
refers to the collection of random variables $\left\{\sum_{i=1}^n f(Y_i)\right\}_{f \in \cF}$, 
where $Y_1,\dots,Y_n$ are  (empirical) i.i.d.\ samples from $P$.
In our case, the function class is indexed by the collection of degree-$m$ pseudoexpectations, 
where the function associated to a pseudoexpectation $\pE$ maps $Y$ to $\pE[\iprod{v,Y}^m] / \pE[\|v\|_2^m]$.

\paragraph{Connections to Gaussian Processes}
A fundamental result of Talagrand concerns the special case 
of {\em linear} processes (when $\cF$ is a set of linear functions).
This result, known as generic chaining or majorizing measures, states that for any $s$-subgaussian 
distribution $P$ over $\R^{k}$ with zero mean and any fixed set $T \subset \R^k$, 
\begin{align}
\label{eq:gen-ch-tal}
\E_{Y \sim P}\left[\sup_{t \in T} \langle t, Y \rangle \right] \lesssim s \cdot \E_{X \sim \cN(0,\bI_k)}\left[\sup_{t \in T} \langle t, X \rangle \right]\,.
\end{align}  
That is, we can bound from above the (expected) behavior of the supremum 
of \emph{linear} functions with subgaussian inputs 
with the analogous behavior under Gaussian inputs.

In fact, even more is true.
Talagrand's result even shows concentration bounds for the supremum, 
captured by the moment bound for any even $p \in \N$, namely,
\begin{align}
\label{eq:gen-ch-tal-conc}
    \left ( \E_{Y \sim P} \left[\sup_{t \in T} \langle t, Y \rangle^p \right] \right )^{1/p} \lesssim s \cdot \sqrt{p} \cdot \E_{X \sim \cN(0,\bI_k)}\left[\sup_{t \in T} \langle t, X \rangle \right] \, ,
\end{align}
whenever $T$ is symmetric.

Our goal is to use generic chaining to bound the empirical process 
in \eqref{eq:our-tech-bdd-moments-samples}.
Trying to massage \eqref{eq:our-tech-bdd-moments-samples} into 
the form of \eqref{eq:gen-ch-tal}, first observe that if $P$ is subgaussian, 
then the concatenated vector $Y' = (Y_1,\dots,Y_n) \in \R^{nd}$ of $n$ independent samples 
from $P$ is also subgaussian.
So, the left-hand side of \eqref{eq:our-tech-bdd-moments-samples} 
is a supremum over a collection of degree-$m$ polynomials 
of the subgaussian random variable $Y'$.
If there were an analogue of generic chaining for degree-$m$ polynomials, 
we would then be able to replace $Y'$ with an analogous Gaussian, 
after which we would be in good shape---indeed, certifiable subgaussianity 
of the Gaussian distribution itself means that 
if we replace $Y_1,\ldots,Y_n$ in \eqref{eq:our-tech-bdd-moments-samples} 
with samples $X_1,\ldots,X_n \sim \cN(0,\bI_d)$, for large-enough $n$, we get
\begin{align}
\label{eq:our-techniques-gaussian-bound}
 \E_{X_1,\dots,X_n \stackrel{\text{i.i.d.}}{\sim} \cN(0,\bI_d)}\left[\sup_{\pE}  \frac{\pE\left[\frac{1}{n}\sum_{i=1}^n\langle v, X_i\rangle^m \right]}{\pE [\|v\|_2^m]}  \right] \leq (C \sqrt m)^m \, .
\end{align} 
Unfortunately, no degree-$m$ variant of \Cref{eq:gen-ch-tal} is possible:
\Cref{eq:gen-ch-tal} is known to fail for simple non-linear functions, 
e.g., degree-two polynomials~\cite{Talagrand21}.
Our main technical innovation lies in enabling 
the use of \Cref{eq:gen-ch-tal} for the specific non-linearity 
involved in our setting.

\paragraph{Linearizing the Nonlinear Process}
The key idea is to replace \Cref{eq:our-tech-bdd-moments-samples} 
with an equivalent linear empirical process.
Suppose that $\pE$ is a degree-$m$ pseudoexpectation 
in variables $v_1,\ldots,v_d$, with $m$ even.
Observe that $\pE \sum_{i=1}^n \iprod{v,Y_i}^m$ is akin to the $\ell_m$-norm, 
raised to the power $m$, of the vector whose entries are $\iprod{v,Y_i}$.
Recall that by H\"older duality of norms, the $\ell_m$-norm of a vector $x$ 
can be written as a linear optimization problem over a convex set:
\begin{align}
\label{eq:our-techniques-holder}
\|x\|_{\ell_m} = \sup_{\|w\|_{\ell_{\frac{m}{m-1}}} \leq 1} \iprod{x,w} \, .
\end{align}
Our plan is to ``lift'' this equality into pseudoexpectations.\footnote{We thank 
David Steurer for pointing out this interpretation of our proof.}

We say that a degree-$m$ pseudoexpectation $\pE'$ in $d+n$ variables $v_1,\ldots,v_d,w_1,\ldots,w_n$ 
\emph{extends} $\pE$, if $\pE' p(v) = \pE p(v)$ for every polynomial in the variables $v$.
Then, using that H\"older's inequality has a simple sum of squares proof, 
we can obtain the following analogue of~\Cref{eq:our-techniques-holder} 
(see \Cref{fact:peholder}):
\begin{align}
\label{eq:our-techniques-pseudo-holder}
\Paren{\pE \left [ \sum_{i = 1}^n \iprod{v,Y_i}^m \right ] }^{1/m} = \sup_{\pE' \text{ extends } \pE} \, \, \, \frac{\pE' \left [ \sum_{i=1}^n \iprod{v,Y_i} w_i^{m-1} \right ] }{\Paren{\pE' \left [ \sum_{i =1}^n w_i^m \right ]}^{1-1/m} } \;.
\end{align} 
Applying \Cref{eq:our-techniques-pseudo-holder} to the left-hand side 
of \Cref{eq:our-tech-bdd-moments-samples}, we get
\[
\E_{Y_1,\dots,Y_n \stackrel{\text{i.i.d.}}{\sim} P}\left[\sup_{\pE}  \frac{\pE\left[\frac{1}{n}\sum_{i=1}^n\langle v, Y_i\rangle^m \right]}{\pE [\|v\|_2]^m}  \right]
= \E_{Y_1,\dots,Y_n \stackrel{\text{i.i.d.}}{\sim} P} \sup_{\pE'} \Paren{ \frac{ \pE' \left [ \sum_{i=1}^n \iprod{v,Y_i} w_i^{m-1} \right ]}{ \Paren{\pE' \left [ \sum_{i =1}^n w_i^m \right ]}^{1-1/m} \cdot \Paren{ \pE' \left [ \|v\|_2^m \right] }^{1/m}} }^m \, . 
\]
Now we have expressed our nonlinear empirical process in the form 
of the $m$-th power of a linear empirical process.
This implies we can apply the concentration bound~\Cref{eq:gen-ch-tal-conc} 
to replace the samples $Y_1,\ldots,Y_n \sim P$ with Gaussian samples 
$X_1,\ldots,X_n \sim \cN(0,\bI)$, to obtain 
\begin{align}
& \E_{Y_1,\dots,Y_n \stackrel{\text{i.i.d.}}{\sim} P}\left[\sup_{\pE}  \frac{\pE\left[\frac{1}{n}\sum_{i=1}^n\langle v, Y_i\rangle^m \right]}{\pE [\|v\|]^m}  \right] \nonumber\\
& \quad \leq (Cs\sqrt{m})^m \cdot \Paren{ \E_{X_1,\dots,X_n \stackrel{\text{i.i.d.}}{\sim} \cN(0,\bI)} \sup_{\pE'} \frac{ \pE' \left [ \sum_{i=1}^n \iprod{v,X_i} w_i^{m-1} \right ]}{ \Paren{\pE' \left [ \sum_{i =1}^n w_i^m \right ]}^{1-1/m} \cdot \Paren{ \pE' \left [ \|v\|_2^m \right] }^{1/m}} }^m  \, ,
\label{eq:our-techniques-getting-to-gaussian}
\end{align}
where $C$ is some universal constant.

\paragraph{Bounding the Gaussian Process}
To complete the proof, we need to analyze the Gaussian process 
whose supremum is represented on the right-hand side of \Cref{eq:our-techniques-getting-to-gaussian}.
Using \Cref{eq:our-techniques-pseudo-holder} followed by \Cref{eq:our-techniques-gaussian-bound}, 
the right-hand side is equal to
\[
(C s \sqrt{m})^m \cdot \Paren{ \E_{X_1,\dots,X_n \stackrel{\text{i.i.d.}}{\sim} \cN(0,\bI)}\left[\sup_{\pE}  \frac{\pE\left[\frac{1}{n}\sum_{i=1}^n\langle v, X_i\rangle^m \right]}{\pE \|v\|_2^m}  \right] } \leq (C' sm)^m \;,
\]
for $C'$ some universal constant, which is what we wanted to show.

%% file: applications.tex
\section{Algorithmic Implications of \Cref{thm:certifiability}} 
\label{sec:apps-detailed}
\label{sec:algo-implications}
In this section, we state the 
algorithmic applications of our main result, some of which were summarized in  \Cref{table:summary}.

We begin by  defining the strong contamination model for the outliers that has become a standard in algorithmic robust statistics~\cite{DiaKKLMS16-focs,DiaKan22-book}. 
\begin{definition}[Strong Contamination Model]
Given a \emph{corruption} parameter $\eps \in (0,1)$ 
and a distribution $P$ on uncorrupted samples, 
an algorithm takes samples from $P$ with \emph{$\eps$-contamination} 
as follows: 
(i) The algorithm specifies the number $n$ of samples it requires. 
(ii) Then $n$ i.i.d.\ samples from $P$ are drawn but not yet shown to the algorithm. These 
samples are called inliers. 
(iii) An arbitrarily powerful adversary then inspects 
the entirety of the $n$ i.i.d.\ samples, 
before deciding to replace any subset of $\lceil \eps n \rceil$ 
samples with arbitrarily corrupted points, 
and returning the modified set of $n$ samples to the algorithm.    
We call the resulting set of samples to be $\eps$-corrupted.
\end{definition}

The rest of this section is organized as follows:
\Cref{ssec:mean-est} focuses on robust mean estimation in Euclidean norm,
\Cref{ssec:sparse-mean} focuses on robust sparse estimation,
\Cref{ssec:cov-est} focuses on covariance estimation and mean estimation in Mahalanobis 
norm, \Cref{ssec:linear-regression} focuses on linear regression, and finally 
\Cref{ssec:clustering} focuses on clustering of mixture models.

\subsection{Outlier-Robust Mean Estimation in Euclidean Norm}
\label{ssec:mean-est}
We begin by considering the problem of mean estimation of $\eps$-corrupted samples 
from a subgaussian distribution.
First, we focus on the regime when the outliers are in minority, 
i.e., $\eps < c<1/2$ for an absolute constant $c$.
When the inliers are sampled from an $s$-subgaussian distribution, 
the optimal error is $\Theta(s\eps\sqrt{\log(1/\eps)})$, 
achievable with sample complexity $\widetilde\Theta(d/\eps^2)$.
However, the existing algorithms that run in polynomial time 
only achieve error $\Omega(\sqrt{\eps})$, 
even if allowed a large degree polynomial $\poly(d/\eps)$ sample size
~\cite{DiaKKLMS17,SteCV18}.  
We break this $\sqrt{\eps}$ barrier for subgaussian distributions, achieving a smooth accuracy-sample-runtime tradeoff quantified by a parameter $t$.

\begin{theorem}[Robust Mean Estimation: Consequence of~\cite{HopLi18,KotSS18} and \Cref{thm:certifiability}]
\label{cor:robust-euclidean}
Let $c<1/2$ be a sufficiently small absolute constant and let $\eps \in (0,c)$.
Let $P$ be an $s$-subgaussian distribution over $\R^d$ with mean $\mu$.
Let $S$ be a set of $\eps$-corrupted samples from $P$ with $n:= |S|$.
For any $t = 2^j$ for some $j\in \N$,
if $n \gtrsim \poly(d^t,1/\eps)$, 
there is an algorithm that (i) takes as inputs $S$, $t$, $\eps$, and $s$, (ii) 
runs in  $(nd)^{\poly(t)}$ time, and (iii) outputs $\widehat{\mu} \in \R^d$ such that, with probability at least $0.9$, $\|\widehat{\mu}-\mu\|_2 \lesssim s \sqrt{t}\eps^{1- \frac{1}{t}}$. 
\end{theorem}
\looseness=-1 Observe that the algorithm of \Cref{cor:robust-euclidean} 
uses $\poly(d^t)$ samples to achieve $\eps^{1-1/t}$ error, 
which is higher than the information-theoretic sample complexity 
(linear in the dimension $d$).
As mentioned in \Cref{rem:tradeoffs-robust-mean}, \cite{DiaKKPP22-colt} 
showed that the above sample complexity is qualitatively tight 
within a broad family of algorithms (SQ algorithms) 
for (hypercontractive) subgaussian inliers.

We next consider the regime of $\eps \in (1/2,1)$.
Since the outliers could be in majority, it is no longer possible 
to estimate the true mean accurately with a single hypothesis, 
motivating the notion of list-decodable learning~\cite{ChaSV17}.
In the list-decodable setting,
the algorithm is allowed to output a small list of vectors, 
say of size $O(1/\alpha))$, such that at least one of them 
is close to the true mean.
Existing polynomial time algorithms for list-decodable subgaussian mean estimation 
were stuck at $\Omega(\sqrt{1/\alpha})$ error~\cite{ChaSV17}, 
despite the information-theoretic error being $\sqrt{\log(1/\alpha)}$, 
achievable with $\Theta_\alpha(d)$ samples~\cite{DiaKS18-list}.
Combining \Cref{thm:certifiability} with the results of \cite{KotSS18}, 
we significantly improve on the $\sqrt{1/\alpha}$ error:
\begin{theorem}[List-decodable Mean Estimation: Consequence of~\cite{KotSS18} and \Cref{thm:certifiability}]
\label{cor:list-decodable-euclidean}
Consider the setting of \Cref{cor:robust-euclidean} and  define $\alpha:=1-\eps$.
If $n \gtrsim \poly(d^t,1/\alpha)$, 
there is an algorithm that (i) takes as inputs $S$, $t$, $\eps$, and $s$, 
(ii) runs in  $(nd)^{\poly(t)}$ time, 
and (iii)  outputs a list $\cL$ of $O(\frac{1}{\alpha})$-vectors such that, 
with probability at least $0.9$, there exists $\widehat{\mu} \in \cL$ 
satisfying $\|\widehat{\mu}-\mu\|_2 \lesssim s \sqrt{t}\alpha^{- \frac{1}{t}}$.     
\end{theorem}
Note that the sample complexity of the algorithm above is larger 
than the information-theoretic sample complexity of $\Theta_\alpha(d)$. 
\cite{DiaKS18-list,DiaKRS23-sq} showed that the above sample bound 
is qualitatively optimal for a broad family of efficient algorithms, 
even when the inliers follow the Gaussian distribution.

\subsection{Outlier-Robust Sparse Mean Estimation }
\label{ssec:sparse-mean}

While the results of the previous section concern \emph{unstructured} 
mean estimation (the true parameter $\mu$ is potentially arbitrary), 
in several applications, the true signal $\mu$ may be structured.
An important structure in statistics is that of \emph{sparsity}: 
the mean $\mu$ is $k$-sparse, i.e., at most $k$ entries of $\mu$ 
are non-zero (importantly, the subset of $[d]$ corresponding to the 
nonzero entries is unknown to the learner).

In a similar vein to the unstructured setting, 
the information-theoretically optimal error is 
$\Theta(\eps\sqrt{\log(1/\eps)})$ and $\Theta(\sqrt{\log(1/(1-\eps))})$, 
for the minority of outliers and list-decodable settings, respectively.
Importantly, the information-theoretic sample complexity is now only 
$\Theta_{\eps,\alpha}(k \log d)$---as opposed to $\Theta(d)$ in the setting of the 
previous section. 
However, existing efficient algorithms for sparse mean estimation 
in the presence of outliers were similarly stuck at errors 
$\sqrt{\eps}$ and $\sqrt{1/\alpha}$~\cite{BalDLS17,DiaKKPP22-colt,DiaKKPP22-neurips}.

Combining the results of \cite{DiaKKPP22-colt,DiaKKPP22-neurips} with \Cref{thm:certifiability}, we obtain the following 
new algorithmic guarantees for robust sparse mean estimation.
\begin{theorem}[Robust Sparse Mean Estimation: Consequence of~\cite{DiaKKPP22-colt,DiaKKPP22-neurips} and \Cref{thm:certifiability}]
\label{lem:sparse-mean-estimation}
Consider the setting of \Cref{cor:robust-euclidean}, with the only difference 
that $\mu$ is $k$-sparse, and define $\alpha:=1-\eps$.
Let $S$ be a set of $\eps$-corrupted samples from $P$ with $n:= |S|$.
For any $t = 2^j$ for some $j\in \N$,
if $n \gtrsim \poly(k^t,1/\eps,1/\alpha,\log d)$, 
there is an algorithm that (i) takes as inputs $S$, $t$, $\eps$, $r$, 
(ii) runs in  $(nd)^{\poly(t)}$ time, and (iii) satisfies the following guarantees:
\begin{itemizec}
    \item (Minority of outliers) If $\eps \in (0,c)$, the algorithm outputs 
    $\widehat{\mu} \in \R^d$ such that, with probability at least $0.9$, 
    $\|\widehat{\mu}-\mu\|_2 \lesssim s \sqrt{t}\eps^{1- \frac{1}{t}}$. 
    \item (Majority of outliers)
    If $\eps \in [c,1)$, the algorithm outputs a list $\cL$ of 
    $O(\frac{1}{\alpha})$-vectors such that, with probability at least $0.9$, 
    there exists $\widehat{\mu} \in \cL$ satisfying 
    $\|\widehat{\mu}-\mu\|_2 \lesssim s \sqrt{t}\alpha^{- \frac{1}{t}}$.     
\end{itemizec}
\end{theorem}
As in the dense setting, the algorithm above obtains qualitatively the optimal accuracy-
sample complexity tradeoff for the aforementioned broad families of 
algorithms~\cite{DiaKKPP22-colt,DiaKKPP22-neurips}.

\subsection{Outlier-Robust Mean and Covariance Estimation}
\label{ssec:cov-est}

We next consider the tasks of mean estimation in Mahalanobis norm 
and covariance estimation in relative spectral norm.
For the problem of mean estimation in Mahalanobis norm,
the goal is to ensure that $\|\Sigma^{-1/2}(\widehat{\mu} - \mu)\|_2$ is small, 
where $\Sigma$ is the \emph{unknown} covariance matrix of the inliers.
For both of these estimation tasks, we consider the case when the inliers follow a 
hypercontractive subgaussian distribution (\Cref{def:hypercontractive}).
The assumption of hypercontractivity is required to avoid the dependence on the condition 
number of the covariance matrix.
The information-theoretic optimal error for both problems 
is $\widetilde{\Theta}(\eps)$, and is attainable 
with $\Theta_\eps(d)$ samples~\cite{KotSS18}.

As opposed to the previous subsections where existing algorithms 
achieved non-trivial error, 
the algorithmic landscape of robust covariance-aware estimation 
for a generic hypercontractive subgaussian distribution is much less understood.
In particular, none of the existing algorithms were known 
to satisfy the following guarantees: (i) robustness to a constant fraction of corruption, 
and (ii) error guarantee  that does not scale polynomially 
with the condition number of the covariance. 
Combining \cite{KotSS18} with \Cref{thm:certifiability},
we obtain the following result:

\begin{theorem}[Robust Mean and Covariance Estimation in Mahalanobis Norm: 
Consequence of~\cite{KotSS18} and \Cref{thm:certifiability-hyper}]
\label{cor:cov-estimation}
There exist absolute constants $c,c'<1/2$ such that the following holds. Let $\eps \in (0,c)$ and let $P$ be an $s$-hypercontractive subgaussian distribution (\Cref{def:hypercontractive}) over $\R^d$ with mean $\mu$ and covariance $\vec \Sigma$.
Let $S$ be a set of $\eps$-corrupted samples from $P$ with $n:= |S|$.
Fix any $t \in \N$ such that $t = 2^j$, for some $j\in \N$, 
and  $s t \eps^{1- \frac{2}{t}} \leq c'$.
If $n \gtrsim \poly(d^t,1/\eps)$, 
there is an algorithm that (i) takes as inputs $S$, $t$, $\eps$, $s$, (ii) 
runs in  $(nd)^{\poly(t)}$ time, and (iii) outputs $\widehat{\mu}$, $\widehat{\vec \Sigma}$ with the following guarantees:
    \begin{itemizec}
        \item Mean Estimation in Euclidean norm: $\|\left(\widehat{\mu} - \mu\right)\|_2\lesssim  s \sqrt{t} \eps^{1-1/t} \|\vec\Sigma\|_\op$.
        \item Mean Estimation in Mahalanobis norm: $\|\vec\Sigma^{-1/2}\left(\widehat{\mu} - \mu\right)\|_2 \lesssim  s \sqrt{t} \eps^{1-1/t}$.
        \item Covariance estimation in relative spectral norm: $(1-\delta)\vec\Sigma \preceq \widehat{\vec\Sigma} \preceq(1+\delta) \vec\Sigma$ for $\delta \lesssim st \eps^{1-\frac{2}{t}}$.
    \end{itemizec}
\end{theorem}
As in prior cases, the accuracy-sample-time tradeoff achieved 
by the algorithm above is qualitatively optimal for SQ algorithms. 
Specifically, the optimality of the first two guarantees directly
follows from the SQ Lower bound of \cite{DiaKKPP22-colt}, 
because the covariance matrix in the provided hard instance is well-conditioned .
The third condition also follows from \cite{DiaKKPP22-colt}, 
because the aforementioned hard instance is also spectrally separated: 
the variance of the hidden direction in the aforementioned hard instance 
is $1 - \Theta(\eps^{1-\frac{2}{t}})$, as opposed to $1$ in the null case.

\begin{remark}[Related Robust Estimation Tasks]  
We remark that robustly estimating the covariance in 
Frobenius norm (with dimension-independent error) is information-theoretically 
impossible for generic hypercontractive subgaussian distributions; 
see, e.g.,~\cite{KotSS18}. 
Existing algorithms for robust covariance estimation in Frobenius norm 
require that the moments of degree-two polynomials are suitably 
bounded (as opposed to degree-one polynomials in \Cref{def:sg}).
Similarly, list-decodable covariance estimation (robustness to larger $\eps$) 
requires stronger properties than subgaussianity~\cite{KarKK19}. 
In fact, 
even for Gaussian marginals, there are inherent obstacles 
to achieve fully polynomial runtime~\cite{DiaKPPS21}.
\end{remark}

\subsection{Outlier-Robust Linear Regression}
\label{ssec:linear-regression}

We next consider the problem of linear regression, 
where the algorithm observes (corrupted) samples $(X,y)$ 
and the goal is to find a linear predictor that approximately 
minimizes the (average) squared loss over the inliers.
Existing algorithms for a generic (hypercontractive) subgaussian distribution 
incur polynomial dependence on the condition number of the covariance matrix 
of the covariates $X$~\cite{DiaKS19,DiaKKLSS19,PraSBR20,PenJL20,CheATJFB20,JamLST21}, which is undesirable in many applications.
Combining the algorithmic framework of \cite{BakPra21} 
with \Cref{thm:certifiability}, we obtain the following theorem:
\begin{theorem}[Robust Linear Regression: Consequence of~\cite{BakPra21} and \Cref{thm:certifiability-hyper}]
\label{cor:rob-regression}
Let $\cD$ be a distribution over $(X,y)$ with $X \in \R^d$ and $y\in\R$ such that $\cD$ is $s$-hypercontractive-subgaussian. 
Define $\ell:\R^D \to \R_+$ as $\ell(\beta):= \E_{(X,y)\sim \cD}\left[(y - \langle \beta,x \rangle)^2\right]$ to be the population loss and    $\beta^* := \arg\min_{\beta \in \R^d }\ell(\beta)$ as the optimal linear predictor on $\cD$.
Let $S$ be a set of $\eps$-corrupted samples from $\cD$ with $n:= |S|$.
\begin{itemizec}[leftmargin=*]
    \item (Case I: Arbitrary noise)
Fix any $t = 2^j$, for some $j\in \N$, 
such that $s t \eps^{1- \frac{2}{t}} \leq c'$.
If $n \gtrsim \poly(d^t,1/\eps)$ for some $\gamma \in (0,1)$, 
there is an algorithm that (i) takes as inputs $S$, $t$, $\eps$, $s$, $\gamma$, 
(ii) runs in $(nd)^{\poly(t)}$-time,  and (iii) outputs an 
estimate $\widehat{\beta}$ such that: 
\begin{itemize}
    \item $\|\vec \Sigma^{1/2} (\widehat{\beta}-\beta^*) \|_2 \lesssim s \sqrt{t}\eps^{1-\frac{2}{k}} \sqrt{\ell(\beta^*)}$
    \item $\ell(\widehat{\beta})  - \ell(\beta^*) \lesssim s^2 t \eps^{2 - \frac{4}{k}} \ell(\beta^*)$.
\end{itemize}
\item 
(Case II: Independent noise)
Assume that $y-X^\top\beta^*$ is independent of $X$ under $\cD$.
Fix any $t = 2^j$, for some $j\in \N$, 
such that $s t \eps^{1- \frac{2}{t}} \leq c'$.
If $n \gtrsim \poly(d^t,1/\eps)$ for some $\gamma \in (0,1)$, 
there is an algorithm that (i) takes as inputs $S$, $t$, $\eps$, $s$, $\gamma$, 
(ii) runs in $(nd)^{\poly(t)}$-time,  and 
(iii) outputs an estimate $\widehat{\beta}$ such that 
\begin{itemize}
    \item $\|\vec \Sigma^{1/2} (\widehat{\beta}-\beta^*) \|_2 \lesssim s \sqrt{t}\eps^{1-\frac{1}{k}} \sqrt{\ell(\beta^*)}$
    \item $\ell(\widehat{\beta})  - \ell(\beta^*) \lesssim s^2 t \eps^{2 - \frac{2}{k}} \ell(\beta^*)$.
\end{itemize}
\end{itemizec}
\end{theorem}
Observe that the result above does not apply to list-decodable linear regression. This is because 
(certifiable) subgaussianity is not sufficient for list-decodable linear regression~\cite{KarKK19}. Similarly to list-decodable covariance estimation, 
even for Gaussian marginals, there are inherent obstacles 
to achieve fully polynomial runtime~\cite{DiaKPPS21}.

\subsection{Clustering Mixture Models} 
\label{ssec:clustering}
\label{ssec:mixtures}
Finally, we study the problem of clustering and/or 
parameter estimation for mixture models. Here the algorithm 
observes ($\eps$-corrupted) samples from $\sum_{i=1}^k \alpha_i P_i$, 
where $\alpha_i \geq 0$ denotes the mixing weights and $P_i$ denotes 
the individual component distributions.
We consider the regime where  $\eps \ll \min_i\alpha_i$, 
so that the outliers do not completely remove a cluster.  
In fact, even for clean data ($\eps = 0$), our \Cref{thm:clustering} is new. 
A prototypical statistical task is parameter estimation, 
where the goal is  estimate 
the mean of individual components 
assuming that their means are pairwise separated. 
To be precise, let $\Delta:= \min_{i \neq j}\|\E_{P_i}[X] - \E_{P_j}[X]\|_2$ 
denote the separation radius.
While $\Delta \gg \sqrt{\log k}$ is information-theoretically sufficient (and necessary) 
to estimate $\E_{X \sim P_i}[X]$~\cite{RegVij17}, current algorithms for generic 
subgaussian distributions need a much larger separation radius, $\Delta \gtrsim k^{\Omega(1)}$.
Combining \Cref{thm:certifiability} with \cite{HopLi18,KotSS18}, we obtain the following 
algorithmic guarantee:
\begin{theorem}[Clustering Subgaussian Mixtures: Consequence of~\cite{HopLi18,KotSS18} 
and \Cref{thm:certifiability}]
\label{thm:clustering}
Let $P_1,\dots,P_k$ be $s$-subgaussian distributions over $\R^d$, with means $\mu_i = \E_{P_i}[X]$ for $i \in [k]$.
Let $\alpha_1,\dots,\alpha_k$ be the mixture weights with $\alpha := \min_i \alpha_i$ and let $\eps$ denote the contamination parameter with $\eps \ll \alpha$.
Let $P$ denote the mixture distribution $\sum_i\alpha_i P_i$ 
and $\Delta$ denote the minimum pairwise separation, i.e., 
$\Delta = \min_{i \neq j} \|\mu_i-\mu_j\|_2$.
Suppose that $\Delta  \gtrsim k^{\gamma}$ for any constant $\gamma>0$.
Let $S$ be a set of $\eps$-corrupted 
$n$ samples from $P$. %
If $n \gg  \poly(d^{1/\gamma},k,1/\alpha)$,
then there is an algorithm that takes as input $S, \alpha, s, \Delta$, 
runs in time $(nd)^{O(1/\gamma)}$, 
and outputs $\widehat{\mu}_1,\dots, \widehat{\mu}_k$ such that with probability at least $9/10$: for all $j\in[k]$, $\min_{i \in [k]}\|\widehat{\mu}_i - {\mu}_j\|_2 \lesssim \poly_s(1/k)$.  
\end{theorem}

We remark that the sample-time tradeoff achieved by the algorithm above 
is qualitatively optimal within SQ algorithms and low-degree polynomial 
tests~\cite{DiaKPZ23-sq-mixture}. The corresponding information-computation 
tradeoff established in \cite{DiaKPZ23-sq-mixture} applies 
even when each $P_i$ is a bounded covariance Gaussian, 
i.e.,  $P_i=\cN(\mu_i, \vec\Sigma)$ for 
an unknown $\vec\Sigma$ with $\vec\Sigma \preceq \bI_d$.
That is, it turns out that 
essentially the hardest instance for the task of clustering subgaussian mixtures 
corresponds to the special case of Gaussian components (with common but unknown covariance).